\documentclass[11pt,a4paper,reqno]{article}
\usepackage{fullpage}
\usepackage{amssymb,latexsym}
\usepackage{graphicx,amsmath,amssymb}
\usepackage{amsthm}
\usepackage{microtype}
\usepackage{enumitem}
\usepackage{tikz}
\usepackage{tabu} 
\usepackage[pdfstartview=FitH,pdfpagemode=None,
           backref,colorlinks=true,citecolor=mycyan,
           linkcolor=mycyan]{hyperref} 
           \hypersetup{
 urlcolor=[rgb]{0.8, .2, .5}
}

\usepackage{setspace}
\usepackage{mathrsfs}
\usepackage{color}
\usepackage{iddo}
\date{}
        \usepackage{stmaryrd}
        \usepackage{tikz}
        \usepackage{tocbibind}
        \usepackage{version}
        \usepackage{xparse}
        \usepackage{xspace}
        \usepackage{xstring} 

        \usepackage{algorithm}

        \makeatletter
        \global\let\tikz@ensure@dollar@catcode=\relax
        \makeatother

        \SetSymbolFont{stmry}{bold}{U}{stmry}{m}{n}

        \PackageWarning{miforbes}{overfullrule enabled}
        \newcommand{\ignore}[1]{}

        \newcommand*{\subproofname}{Sub-Proof:}

        \PackageWarning{miforbes}{get rid of matrix commands}

        \renewcommand{\vec}[1]{\overline{#1}}

        \newcommand{\cG}{\mathcal{G}}

        \newcommand{\cvG}{{\vec{\cG}}}

        \makeatletter
        \newcommand{\va}{{\vec{a}}\@ifnextchar{^}{\!\:}{}}
        \newcommand{\vc}{{\vec{c}}\@ifnextchar{^}{\!\:}{}}
        \newcommand{\vd}{{\vec{d}}\@ifnextchar{^}{\!\:}{}}
        \newcommand{\ve}{{\vec{e}}\@ifnextchar{^}{\!\:}{}}
        
        \newcommand{\vg}{{\vec{g}}\@ifnextchar{^}{\!\:}{}}
        \newcommand{\vh}{{\vec{h}}\@ifnextchar{^}{\!\:}{}}
        \newcommand{\vi}{{\vec{i}}\@ifnextchar{^}{\!\:}{}}
        \newcommand{\vj}{{\vec{j}}\@ifnextchar{^}{\!\:}{}}
        \newcommand{\vk}{{\vec{k}}\@ifnextchar{^}{\!\:}{}}
        \newcommand{\vl}{{\vec{\ell}}\@ifnextchar{^}{\!\:}{}}
        \newcommand{\vm}{{\vec{m}}\@ifnextchar{^}{\!\:}{}}
        \newcommand{\vn}{{\vec{n}}\@ifnextchar{^}{\!\:}{}}
        \newcommand{\vo}{{\vec{o}}\@ifnextchar{^}{\!\:}{}}
        \newcommand{\vp}{{\vec{p}}\@ifnextchar{^}{\!\:}{}}
        \newcommand{\vq}{{\vec{q}}\@ifnextchar{^}{\!\:}{}}
        \newcommand{\vr}{{\vec{r}}\@ifnextchar{^}{\!\:}{}}
        \newcommand{\vs}{{\vec{s}}\@ifnextchar{^}{\!\:}{}}
        \newcommand{\vt}{{\vec{t}}\@ifnextchar{^}{\!\:}{}}
        \newcommand{\vu}{{\vec{u}}\@ifnextchar{^}{\!\:}{}}
        \newcommand{\vv}{{\vec{v}}\@ifnextchar{^}{\!\:}{}}
        \newcommand{\vw}{{\vec{w}}\@ifnextchar{^}{\!\:}{}}
        \newcommand{\vy}{{\vec{y}}\@ifnextchar{^}{\!\:}{}}
        \newcommand{\vz}{{\vec{z}}\@ifnextchar{^}{\!\:}{}}

        \newcommand{\vA}{{\vec{A}}\@ifnextchar{^}{\!\:}{}}
        \newcommand{\vB}{{\vec{B}}\@ifnextchar{^}{\!\:}{}}
        \newcommand{\vC}{{\vec{C}}\@ifnextchar{^}{\!\:}{}}
        \newcommand{\vD}{{\vec{D}}\@ifnextchar{^}{\!\:}{}}
        \newcommand{\vE}{{\vec{E}}\@ifnextchar{^}{\!\:}{}}
        \newcommand{\vF}{{\vec{F}}\@ifnextchar{^}{\!\:}{}}
        \newcommand{\vG}{{\vec{G}}\@ifnextchar{^}{\!\:}{}}
        \newcommand{\vH}{{\vec{H}}\@ifnextchar{^}{\!\:}{}}
        \newcommand{\vI}{{\vec{I}}\@ifnextchar{^}{\!\:}{}}
        \newcommand{\vJ}{{\vec{J}}\@ifnextchar{^}{\!\:}{}}
        \newcommand{\vK}{{\vec{K}}\@ifnextchar{^}{\!\:}{}}
        \newcommand{\vL}{{\vec{L}}\@ifnextchar{^}{\!\:}{}}
        \newcommand{\vM}{{\vec{M}}\@ifnextchar{^}{\!\:}{}}
        \newcommand{\vN}{{\vec{N}}\@ifnextchar{^}{\!\:}{}}
        \newcommand{\vO}{{\vec{O}}\@ifnextchar{^}{\!\:}{}}
        \newcommand{\vP}{{\vec{P}}\@ifnextchar{^}{\!\:}{}}
        \newcommand{\vQ}{{\vec{Q}}\@ifnextchar{^}{\!\:}{}}
        \newcommand{\vR}{{\vec{R}}\@ifnextchar{^}{\!\:}{}}
        \newcommand{\vS}{{\vec{S}}\@ifnextchar{^}{\!\:}{}}
        \newcommand{\vT}{{\vec{T}}\@ifnextchar{^}{\!\:}{}}
        \newcommand{\vU}{{\vec{U}}\@ifnextchar{^}{\!\:}{}}
        \newcommand{\vV}{{\vec{V}}\@ifnextchar{^}{\!\:}{}}
        \newcommand{\vW}{{\vec{W}}\@ifnextchar{^}{\!\:}{}}
        \newcommand{\vY}{{\vec{Y}}\@ifnextchar{^}{\!\:}{}}
        \newcommand{\vX}{{\vec{X}}\@ifnextchar{^}{}{}}          
        \newcommand{\vZ}{{\vec{Z}}\@ifnextchar{^}{\!\:}{}}
        \makeatother

\newcommand{\demph}[1]{\textbf{\textit{#1}}}

\author{Fedor Part\thanks{Department of Computer Science, Royal Holloway, University of London.  Fedor.Part@gmail.com} \and Iddo Tzameret\thanks{Department of Computer Science, Royal Holloway, University of London. Iddo.Tzameret@rhul.ac.uk} }

\usepackage[normalem]{ulem}
\usepackage{pifont}
\usepackage{bbding}

\usepackage{bussproofs}
\usepackage{microtype}

\newcommand{\iddo}[1]{}
\newcommand{\fedor}[1]{}

\newcommand{\DT}[1]{\text{DT}(#1)}

\newcommand{\DTsw}[1]{\text{DT}_{sw}(#1)}

\newcommand{\Ax}{\Gamma}

\newcommand{\reslin}[1]{{\normalfont Res}$(\text{\normalfont lin}_{#1}${\normalfont)}}

\newcommand{\reslinsem}[1]{\normalfont Sem-Res($\text{\normalfont lin}_{#1}${\normalfont)}}
\newcommand{\reslinsw}[1]{$\text{\normalfont Res}_{sw}(\text{\normalfont lin}_{#1}${\normalfont)}}
\newcommand{\treslin}[1]{{tree-like \normalfont Res($\text{\normalfont lin}_{#1}${\normalfont)}}}
\newcommand{\treslinAx}[2]{{tree-like \normalfont Res(}$\text{\normalfont lin}_{#1},#2${\normalfont)}}

\newcommand{\treslinsw}[1]{$\text{tree-like \normalfont Res}_{sw}(\text{\normalfont lin}_{#1}${\normalfont{)}}}
\newcommand{\PCa}[1]{$\normalfont{\normalfont PC}_{#1}$}

\newcommand{\linSys}[1]{\normalfont{\textsf{LinSys}}(#1)}
\newcommand{\subSum}[1]{\normalfont{\textsf{SubSum}}(#1)}
\newcommand{\negIm}[1]{\normalfont{\textsf{ImAv}}\left(#1\right)}
\newcommand{\imAx}[1]{\normalfont{\textsf{Im}}\!\left(#1\right)\!}
\newcommand{\phpInj}{\normalfont{\textsf{Pigeons}}}
\newcommand{\phpTot}{\normalfont{\textsf{Holes}}}

\newlength{\defbaselineskip}
\renewcommand{\singlespacing}{\setlength{\baselineskip}{1.2\defbaselineskip}}

\begin{document}

\title{Resolution with Counting:\\ Dag-Like Lower Bounds and Different Moduli}

\renewcommand{\singlespacing}{\setlength{\baselineskip}{1.0\defbaselineskip}}

\maketitle

\newcommand{\R}{\ensuremath{R}}

\renewcommand{\check}[1]{\textsf{\textcolor[rgb]{0.501961,0,1}{Check: #1}}}

\newcommand*\circled[1]{\tikz[baseline=(char.base)]{
   \node[shape=circle,draw,inner sep=1pt] (char) {#1};}}

\begin{abstract}
\textit{Resolution over linear equations} is a natural extension of the popular  resolution  refutation system, augmented with the ability to carry out basic counting. Denoted \reslin{\R}, this refutation system operates with disjunctions of linear equations with boolean  variables over a ring \R, to refute unsatisfiable sets of such disjunctions. 
Beginning in the work of \cite{RT07}, through the work of \cite{IS14} which focused on tree-like lower bounds, this refutation system was shown to be fairly strong. Subsequent work (cf.~\cite{Kra16,IS14,KO18,GK17})  made it  evident that establishing lower bounds against general \reslin R\ refutations is a challenging and interesting task since the system captures a ``minimal'' extension of resolution with counting gates for which no super-polynomial lower bounds are known to date. 


We provide the first super-polynomial size  lower bounds on general (dag-like) resolution over linear equations refutations in the  large characteristic regime. 
In particular we prove that the subset-sum principle  $1+ x_1 +\dots +2^n x_n = 0$ requires refutations of exponential-size over \Q. %
Our proof technique is nontrivial and  novel: roughly speaking, we show that under certain conditions every refutation of a subset-sum instance $f=0$ must pass through a fat clause containing an equation $f=\alpha$ for each $\alpha$ in the image of $f$ under boolean assignments. We develop  a somewhat different approach 
to prove exponential lower bounds against tree-like refutations of any subset-sum instance that depends on $n$ variables, hence also separating tree-like from dag-like refutations over the rationals.



We then turn to the finite fields regime, showing that the work of Itsykson and Sokolov \cite{IS14} who  obtained  tree-like lower bounds over $\F_2$ can be carried over and extended to every finite field. We establish new lower bounds and separations as follows: (\textbf{i}) 
for every pair of distinct primes $p,q$, there exist CNF formulas  with short tree-like refutations in \reslin{\F_p} that require exponential-size tree-like  \reslin{\F_q} refutations; (\textbf{ii}) random $k$-CNF formulas require exponential-size tree-like \reslin{\F_p} refutations, for every prime $p$ and constant $k$; and    (\textbf{iii}) exponential-size lower bounds for tree-like \reslin{\F} refutations of the pigeonhole principle, for \emph{every} field \F. 



 


\end{abstract}

 \note\ 
 \textit{The first 10 pages hold a detailed introduction to this work, including background, description of our results and proof techniques}.

\textit{This is  an improved version of a preliminary manuscript that has been circulated before. In particular, tree-like lower bounds on any subset-sum instance that depends on $n$ variables have been added, results about the proof complexity of \emph{linear systems} were added (Sec.~\ref{sec:Linear-Systems-with-Small-Coefficients}), and the dag-like lower bound was modified and rectified to deal with the weakening rule.}


\section{Introduction}

The resolution refutation system is among the most prominent and well-studied propositional proof systems, and for good reasons: it is a  natural and simple refutation system, that, at least in practice, is capable of being easily automatized. Furthermore,  while being non-trivial, it is simple enough to succumb to many lower bound techniques.     

Formally, a resolution refutation of an unsatisfiable CNF formula is a sequence of clauses $D_1,\dots,D_l=\emptyset$, where $\emptyset$ is the empty clause, such that each $D_i$ is either a clause of the CNF or is derived from previous clauses $D_j,D_k,j\leq k<i$ by means of applying the  following \emph{resolution rule}: from the clauses $C\vee x$ and $D\vee\neg x$ derive $C\vee D$. 

The \emph{tree-like} version of resolution, where every occurrence of a clause in the  refutation is used at most once as a premise of a rule, is of
particular importance, since it helps us  to understand certain kind of satisfiability algorithms known as DPLL algorithms (cf.~\cite{Nor15-siglog}). DPLL algorithms are simple recursive algorithms for solving SAT that are the basis of successful contemporary SAT-solvers.
%
The transcript of a run of DPLL on an unsatisfiable formula is a decision tree, which can be interpreted as a tree-like resolution refutation. Thus, lower bounds on the size of tree-like resolution refutations imply lower bounds on the run-time of DPLL algorithms (though it is important to clarify that contemporary SAT-solvers utilize more than the strength of tree-like resolution).

In contrast to the apparent practical success of SAT-solvers, a  variety of hard instances that require exponential-size refutations have been found for resolution during the years. Many classes of such hard instances are based on principles expressing some sort of counting. One famous example is the \textit{pigeonhole principle}, denoted $\text{PHP}^m_n$,  expressing that there is no (total) injective map from a set with cardinality $m$ to a set with cardinality $n$ if $m>n$ \cite{Hak85}. Another important example is \textit{Tseitin tautologies}, denoted $\text{TS}_G$, expressing that the sum of the degrees of vertices in a graph $G$ must be even \cite{Tse68}.

Since such counting tautologies are a source of hard instances for resolution, it is useful to study extensions of resolution that can efficiently count, so to speak. 
This is important firstly, because such systems may become the basis of  more efficient SAT-solvers and secondly, in order  to extend the frontiers of lower bound techniques against stronger and stronger propositional proof systems. Indeed, there are many  works dedicated to the study of weak systems operating with De Morgan formulas with counting connectives; these are variations of resolution that operate with disjunctions of certain arithmetic expressions.
%

One such extension of resolution was introduced by Raz and Tzameret \cite{RT07} under the name \textit{resolution over linear equations}  in which literals are replaced by linear equations. Specifically, the system
R(lin), which operates with disjunctions of linear equations over $\Z$ was studied in \cite{RT07}. This work demonstrated the power of resolution with counting over the integers, and specifically provided polynomial upper bounds for the pigeonhole principle and the Tseitin formulas, as well as other basic counting formulas.  It also established exponential lower bounds for a subsystem of R(lin), denoted $\text{R}^0(\text{lin})$. 
Subsequently, Itsykson and Sokolov \cite{IS14} studied resolution over linear equations over  $\F_2$,  denoted Res($\oplus$). They demonstrated the power of resolution with counting mod 2 as well as its limitations by means of several upper and tree-like lower bounds. Moreover, \cite{IS14} introduced DPLL algorithms, which can ``branch'' on arbitrary linear forms over $\F_2$, as well as parity decision trees, and showed a correspondence between parity decision trees and tree-like Res($\oplus$) refutations. In both \cite{RT07} and \cite{IS14} the dag-like lower bound question for resolution over linear equations remained open.
\medskip

Apart from being a very natural refutation system, understanding the proof complexity of resolution over linear equations is important for the following reason: {proving super-polynomial dag-like lower bounds against resolution over linear equations for prime fields and for the integers can be viewed as a first step towards the long-standing open problems of $\ACZ[p]$-Frege and \TCZ-Frege lower bounds, respectively. We explain this in what follows.  


Resolution operates with clauses, which are De Morgan formulas ($\neg$, unbounded fan-in $\vee$ and $\wedge$) of a particular kind, namely, of depth 1. Thus, from the perspective of proof complexity, resolution is a fairly weak version of the propositional-calculus, where the latter operates with arbitrary De Morgan formulas. Under a natural and general definition,  propositional-calculus systems go under the name \emph{Frege systems}: they can be  (axiomatic) Hilbert-style systems or sequent-calculus style systems. 
The task of proving lower bounds for general Frege systems is notoriously hard: no nontrivial lower bounds are known to date. Basically, the strongest fragment  of Frege systems, for which lower bounds are known are \ACZFrege\ systems, which are Frege proofs  operating with constant-depth formulas. For example, both $\text{PHP}^m_n$ and $\text{TS}_G$ do not admit sub-exponential  proofs in \ACZFrege\ \cite{Ajt88,PBI93,KPW95,BS02}. However, if we extend the De Morgan language with counting connectives such as unbounded fan-in mod $p$ 
($\ACZ[p]$-Frege) or threshold gates (\TCZ-Frege), then we step again into the darkness: proving super-polynomial lower bounds for these systems is a long-standing open problem on what can be characterized as the ``frontiers'' of proof complexity. Recent works by Kraj\'{i}\v{c}ek \cite{Kra16}, Garlik-Ko\l odziejczyk \cite{GK17} and Kraj\'{i}\v{c}ek-Oliveira \cite{KO18} had suggested possible approaches to attack dag-like \reslin{\F_2} lower bounds (though this problem remains open to date). 


\subsection{Our Results and Techniques}\label{intro:our-results}
In this work  we %
prove a host of new lower bounds, separations and upper bounds for resolution over linear equations. Our main  novel technical  contribution is a dag-like refutation lower bound over large characteristic fields. Conceptually, the proof idea exploits two main properties that recently have been found useful in proof complexity:\vspace{-3pt}




\begin{enumerate}
\item[(i)] Single axiom: the hard instance consists of a single unsatisfiable axiom (for boolean assignments) 
\begin{equation}\label{eq:bvp}
1+x_1+\dots +2^n x_n=0
\end{equation}
(unlike, for instance, a set of clauses). \vspace{-6pt} 
\item[(ii)] Large coefficients: the hard instance uses coefficients of exponential magnitude.  
\end{enumerate}
Although employing different approaches, both of these properties played a recent role in proof complexity lower bounds. Forbes et al.~\cite{FSTW16} used  subset-sum variants  (that is, unsatisfiable linear equations with boolean variables) to establish lower bounds on subsystems of the ideal proof system (IPS) over large characteristic fields, where IPS\ is the strong proof system introduced by Grochow and Pitassi \cite{GP14}. It is essential in both \cite{FSTW16} and our work that the hard instance takes the form of a single unsatisfiable axiom. Subsequently, in a very recent work, Alekseev et al.~\cite{AGHT19} established  conditional exponential-size lower bounds on full IPS refutations over the rationals  of the same subset-sum instance \eqref{eq:bvp}, where the use of big coefficients is again essential to the lower bound. We explain our deg-like lower bound in Section \ref{sec:characteristic-zero}.


%
%

\medskip 

The other novel contribution we make is a systematic development of  new kinds of  lower bound techniques against  \emph{tree-like} resolution over linear equations, both over the rationals and over finite fields. To this end we  develop  new and extend existing combinatorial techniques such as the Prover-Delayer game
method as originated in Pudlak and Impagliazzo \cite{PI00} for resolution, and developed further by Itsykson and Sokolov \cite{IS14}. Moreover, we  provide new applications in proof complexity of different combinatorial results; this include bounds on the size of essential coverings of the hypercube from Linial and Radhakrishnan \cite{LINIAL05}, a result about the hyperplane
coverings of the hypercube by Alon and F\"uredi \cite{AF93},  the notion of immunity from Alekhnovich and Razborov \cite{AR01} and Gilbert bound on linear error correcting
codes. 
 We further non-trivially extend the well-established principle of size-width tradeoffs  in resolution \cite{BSW99} to the setting of \reslin R\ (though it is important to note that most of our lower bounds do not follow from this tradeoff result).

%

\subsubsection{Background}
For a ring $R$, the refutation system \reslin{R} is defined as an extension of the resolution refutation system as follows (see Raz and Tzameret \cite{RT07}). The \emph{proof-lines} of \reslin{\R} are called \bemph{linear clauses} (sometimes called simply \emph{clauses}), which are defined as disjunctions of linear equations (with duplicate equations contracted). More formally, they are  disjunctions of the form:
$$
\left(\sum\nolimits_{i=1}^na_{1i}x_i+b_1=0\right)\vee\dots\vee\left(\sum\nolimits_{i=1}^na_{ki}x_i+b_k=0\right),
$$
where $k$ is some number (the \emph{width} of the clause), and $a_{ji}, b_j \in \R$. The \emph{resolution rule} is the following:  
\begin{center}
from  $(C\vee f=0)$ and $(D\vee g=0)$ derive  $(C\vee D\vee (\alpha f + \beta g) = 0),$
\end{center}
where $\alpha,\beta\in \R$, and where $C,D$ are linear clauses. A \reslin{R} \emph{refutation} of an unsatisfiable over 0-1 set of linear clauses $C_1,\ldots,C_m$ is a sequence of proof-lines, where each proof-line is either $C_i$, for $i\in[m]$, a boolean axiom $(x_i=0\vee x_i=1)$ for  some variable $x_i$, or was derived from previous proof-lines by the above resolution rule, or by the \emph{weakening rule} that allows to extend clauses with arbitrary disjuncts, or a \emph{simplification rule} allowing to discard false constant linear forms (e.g., $1=0$) from a linear clause. The last proof-line in a refutation is the empty clause (standing for the truth value \textsf{false}).

The \demph{size} of a \reslin R\  refutation is the total size of all the clauses
in the derivation, where the size of a clause is defined to be the total number of occurrences
of variables in it plus the total size of all the coefficient occurring in the clause. The size of a
coefficient when using integers (or integers embedded in characteristic zero rings) is the
standard size of the binary representation of integers (nevertheless, when we talk about ``big'' or ``exponential''  coefficients and ``polynomially bounded'' coefficients, etc., we mean that the \emph{magnitude} of the coefficients is big (exponential) or polynomially bounded).
 

We are generally interested in the  following questions: 
\begin{itemize}
\item[(Q1)] For a given ring  $\R$, what kind of counting can be efficiently performed in   \reslin{\R} and tree-like \reslin{\R}?
 \item[(Q2)] Can dag-like \reslin{\R} be separated from tree-like \reslin{\R}?
\item[(Q3)] Can tree-like systems for different rings $\R$ be separated?
\end{itemize}

\para{Tree-like \reslin{R} with semantic weakening.} 
In order to be able to do some non-trivial counting in \emph{tree-like} versions of resolution over linear equations we define a semantic version of the system as follows.

The system \reslinsw{\R}  is obtained from \reslin{R} by replacing the weakening and 
the simplification rules, as well as the boolean axioms, with the \emph{semantic weakening} rule (the symbol $\models$ will denote in this work semantic implication \emph{with respect to 0-1 assignments}):\footnotemark
\begin{prooftree}
   \AxiomC{$C$}
   \RightLabel{($C\models D$)\,.}
   \UnaryInfC{$D$}
\end{prooftree}
\footnotetext{Let $k=char(\R)$ be the characteristic of the ring \R. In case $k\notin \{1,2,3\}$, deciding whether an  \R-linear clause $D$ is a tautology (that is, holds for every 0-1 assignment to its variables) is at least as hard as deciding whether a 3-DNF is a tautology (because over characteristic $k\notin \{1,2,3\}$  linear equations can express conjunction of three conjuncts). For this reason \reslinsw{\R} proofs cannot be checked in polynomial time and thus \reslinsw{\R}
is not a Cook-Reckhow proof system unless $\P=\coNP$ (namely, the correctness of proofs in the system cannot necessarily be checked in polynomial-time, as required by a Cook-Reckhow propositional proof system \cite{CR79}; see Section \ref{sec:Propositional-Proof-Systems}).} 

The reason for studying  \reslinsw{\R} is mainly the following: Let $\Ax$ be an arbitrary set of tautological \R-linear clauses. Then, lower bounds for \treslinsw{\R} imply lower bounds for \treslin{\R} with formulas in $\Gamma$ as axioms. For example,
in case \F\ is a field of characteristic 0, the possibility to do counting in \treslin{\F} is quite limited. For instance, we show that  $2x_1+\cdots+2x_n=1$ requires an exponential-size in $n$ refutations (Theorem~\ref{thm:ssTLLB}). On the other hand, such contradictions \emph{do} admit short tree-like \reslin\F\ refutations in the presence of the following \emph{generalized boolean axioms} (which is a tautological linear clause): 
\begin{equation}\label{eq:gen-bool-axiom}
\imAx{f}:=\bigvee\nolimits_{A\in im_2(f)}(f=A),
\end{equation} 
where $im_2(f)$ is the image of $f$ under 0-1 assignments. Similar to the way the boolean  axioms $(x_i=0)\lor(x_i=1)$ state that the possible value
of a variable is either zero or one, the $\imAx{f}$ axiom states
all the possible values that the linear form $f$ can have. If a lower bound holds for \treslinsw{\F} it also holds, in particular, for
\treslin{\F} with the axioms $\imAx{f}$, and this makes \treslinsw{\F} a useful system, for which lower bounds against are sufficiently interesting.
%

\subsubsection{Characteristic Zero Lower Bounds} \label{sec:characteristic-zero}

For characteristic zero fields we will use mainly the rational number field \Q\ (though many of the results hold over any characteristic zero rings). First, we show that over \Q, whenever  $\alpha_1x_1+\cdots+\alpha_nx_n+\beta=0$ is unsatisfiable (over 0-1 assignments),
it has polynomial dag-like \reslin\Q\ refutations if the coefficients are polynomially bounded in magnitude, while it requires exponential dag-like \reslin\Q\ refutations for some subset-sum instances with exponential-magnitude coefficients. Note that $\alpha_1x_1+\cdots+\alpha_nx_n+\beta=0$ expresses
the \emph{subset-sum principle}:   $\alpha_1x_1+\cdots+\alpha_nx_n=-\beta$ is satisfiable iff there is a subset of the integral coefficients
$\alpha_i$ whose sum is precisely $-\beta$.  The lower bound is stated in the following theorem: 
\begin{theorem*}[Theorem~\ref{thm:dagLB}; Main dag-like lower bound]
Any \reslin{\Q} refutation of $x_1+2x_2+\cdots+2^nx_n+1=0$ requires size $2^{\Omega(n)}$.
\end{theorem*}

The proof of this theorem introduces a new lower bound technique. We show that every (dag- or tree-like) refutation $\pi$ of $x_1+2x_2+\cdots+2^nx_n+1=0$ can be transformed without  much increase in size into a derivation of a certain ``fat'' (exponential-size) clause $C_{\pi}$ from boolean axioms only.\footnote{The notion of showing that a refutation must go though a fat (i.e., wide) clause is well established in resolution lower bounds. However,  we  note that our lower bound is completely different from the known size-width based resolution lower bounds (as formulated in a generic way in the work of   Ben-Sasson and Wigderson \cite{BSW99}).} In order to prove that $C_\pi$ is fat, we ensure that every disjunct $g=0$ in $C_{\pi}$ has at most $2^{cn}$ satisfying boolean assignments, for some constant $c<1$. Because $C_{\pi}$ is derived from boolean  axioms alone, it must be a boolean tautology, that is, it must have $2^n$ satisfying assignment. Since every disjunct in $C_\pi$ is satisfied by at most $2^{cn}$ assignments, the number of disjuncts in the clause is at least $2^{(1-c)n}$. Since our constructed derivation is not much larger than the original refutation, the size of the original refutation must be $2^{\Omega(n)}$.

This proof relies in an essential way on the fact that the coefficients of the linear form have exponential  magnitude. Indeed, every contradiction of the form $f=0$
can be shown to admit polynomial-size dag-like \reslin\Q\ refutations whenever the coefficients of $f$ are polynomially bounded. 
A natural question is whether in the case of bounded coefficients, $f=0$ can be efficiently refuted already by \treslin\Q\ refutations.
The question turns out to be non-trivial, and we provide a negative answer:

\begin{theorem*}[Theorem~\ref{thm:ssTLLB}; Subset-sum tree-like lower bounds]
Let $f$ be any linear polynomial over \Q, which depends on $n$ variables. Then \treslin\Q\ refutations of $f=0$ are of size $2^{\Omega(\sqrt{n})}$. 
\end{theorem*}

The proof is in two stages. 
First, we use a transformation analogous to the one used for the dag-like lower bound to reduce the lower bound problem for
refutations of $f=0$ to a lower bound problem for derivations of clauses of a certain kind. Namely, we transform any tree-like refutation
$\pi$ of $f=0$ to a tree-like derivation of $C_{\pi}$ from boolean  axioms without much increase in size. The only difference is that this time we
ensure that in every disjunct $g=0$ of $C_{\pi}$, the linear polynomial $g$ depends on at least $\frac{n}{2}$ variables.

Second, we prove that \treslin\Q\ derivations of such a $C_{\pi}$ are large:


\begin{theorem*}[Theorem~\ref{thm:largeWeightLB}]
Any \treslin{\Q} derivation of any tautology of the form $\bigvee\nolimits_{j\in [N]}g_j=0$, for some positive $N$, where each $g_j$ is linear over $\Q$ and depends on at least $\frac{n}{2}$
variables, is of size $2^{\Omega(\sqrt{n})}$.
\end{theorem*} 

To prove this, as well as some other lower bounds, we extend the Prover-Delayer game technique as originated in Pudlak-Impagliazzo
\cite{PI00} for resolution, and developed further by Itsykson-Sokolov \cite{IS14} for \reslin{\F_2}, to general rings, including characteristic
zero rings (see Sec.~\ref{sec:PD-game}).\footnote{We note here (see Remark 1 in the next sub-section) that the lower bounds that we prove using Prover-Delayer games techniques in case $char(\F)=0$ \emph{do not} follow from lower bounds for \PCa{\F} using size-width relations.} 

We define a non-trivial strategy for Delayer in the corresponding game and prove that it guarantees $\sqrt{n}$ coins using a bound on the size
of essential coverings of the hypercube from Linial and Radhakrishnan \cite{LINIAL05}. The relation between Prover-Delayer games and \treslin\Q\ refutations allows us to conclude that 
the size of \treslin\Q\ refutations must be $2^{\Omega(\sqrt{n})}$.

Moreover, as a corollary of Theorem~\ref{thm:largeWeightLB} we obtain a lower bound on \treslin\Q\ \emph{derivations} (in contrast to refutations) of $\imAx{f}$~:

\begin{corollary*}[Corollary~\ref{cor:imTLLB}]
Let $f$ be any linear polynomial over \Q\ that depends on $n$ variables. Then \treslin\Q\ derivations of $\imAx{f}$ are of size $2^{\Omega(\sqrt{n})}$. 
\end{corollary*}

We also use Prover-Delayer games to prove an exponential-size  $2^{\Omega(n)}$ lower bound on  tree-like 
\reslinsw{\F} refutations of the pigeonhole principle $\text{PHP}^m_n$ \emph{for every field }$\F$ (including finite fields). This extends a previous result by Itsykson and Sokolov \cite{IS14}
for tree-like \reslin{\F_2}. 
\begin{theorem*}[Theorem~\ref{phpLB}; Pigeonhole principle lower bounds]
Let $\F$ be any (possibly finite) field. Then every tree-like \reslinsw{\F} refutation of $\neg{\rm PHP}^m_n$ has size $2^{\Omega\left(\frac{n-1}{2}\right)}$.
\end{theorem*}



Together with the  
polynomial upper bounds for $\text{PHP}^m_n$ refutations in dag-like \reslin{\F} for fields $\F$ of characteristic zero demonstrated by Raz and Tzameret \cite{RT07}, Theorem~\ref{phpLB} 
establishes a \emph{separation between dag-like \reslin{\F} and tree-like \reslinsw{\F}} for characteristic zero fields, for the language of unsatisfiable formulas in CNF:

\begin{corollary*} Over fields of characteristic zero \F, 
 \reslin \F\ has an exponential speed-up over tree-like \reslin \F\, as refutation systems for unsatisfiable formulas in CNF. 
\end{corollary*}

To prove Theorem \ref{phpLB} we need to prove that Delayer's strategy from \cite{IS14} is  successful over any  field.
This argument is new, and uses a result of  Alon-F\" uredi \cite{AF93} about the hyperplane coverings of the hypercube. 
\smallskip
 
We prove another  separation between dag-like \reslin{\Q} and tree-like \reslinsw{\Q}, as follows. For any ring \R\ we define the \emph{image avoidance principle} to be:
$$
\negIm{x_1+\dots+x_n}:=\{\langle x_1+\dots+x_n\neq k\rangle\}_{k\in\{0,\dots,n\}},
$$ 
where $\langle x_1+\dots+x_n\neq k\rangle := \bigvee\nolimits_{k'\in\{0,\dots,n\},~ k\neq k'}x_1+\dots+x_n=k'$. In words, the image avoidance principle expresses the contradictory statement that  for every $0\le i\le n$, $x_1+\dots+x_n$ equals some element in $\{0,\ldots,n\}\setminus i$. In more generality, let $f$ be a  linear form over $\Q$ and let $im_2(f)$ be the image of $f$ under 0-1 assignments to its variables. Define ${\langle f\neq A\rangle}:={\bigvee\nolimits_{A\neq B\in im_2(f)}(f=B)}$, where $A\in \Q$. We define
\begin{equation}\label{eq:imAv}
\negIm{f}:=\{\langle f\neq A\rangle: {A\in im_2(f)} \}\,.
\end{equation}

\begin{corollary*}[Corollary \ref{imAvUB}]
\label{imAvUB}
For every ring \R\ and  every linear form $f$ the contradiction $\negIm{f}$ admits polynomial-size \reslin{R} refutations.
\end{corollary*}


\begin{theorem*}[Theorem~\ref{thm:imAvLB}]
We work over \Q. Let $f=\epsilon_1x_1+\dots+\epsilon_nx_n$, where $\epsilon_i\in\{-1,1\}$. Then 
any tree-like \reslinsw{\Q} refutation of $\negIm{f}$ is of size at least $2^{\frac{n}{4}}$.
%
\end{theorem*}

The lower bound in Theorem \ref{thm:imAvLB} is one more novel application of the Prover-Delayer game argument, combined with the notion of immunity from Alekhnovich and Razborov \cite{AR01}, as we now briefly explain. 

Let $f$ be a linear form as in Theorem \ref{thm:imAvLB}. We consider an instance of the Prover-Delayer game for $\negIm{f}$. 
A position in the game is determined by a \emph{set $\Phi$ of linear non-equalities} of the form $g\neq 0$, which we think of as the set
of non-equalities learned up to this point by Prover. In the beginning $\Phi$ is empty. We define Delayer's strategy in
such a way that for $\Phi$ an end-game position, there is a satisfiable subset $\Phi'=\{g_1\neq 0,\ldots,g_m\neq 0\}\subseteq\Phi$ such that $\Phi'\models f=A$ 
for some $A\in\F$, and Delayer earns at least $|\Phi'|=m$ coins. Because $\F$ is of characteristic zero, it  follows that
${f\equiv A+1~(\text{mod}\ 2)\models} f\neq A\models g_1\cdot\ldots\cdot g_m=0$ and thus the $\frac{n}{4}$-immunity of $f\equiv A+1(\text{mod}\ 2)$ (\cite{AR01})
implies $m\geq \frac{n}{4}$. To conclude, by a standard argument if Delayer always earns $\frac{n}{4}$ coins, then the shortest proof is of size at least
$2^{\frac{n}{4}}$.
\smallskip 

Table \ref{tab:1}  sums up our knowledge up to this point with respect to \Q\ (and for some cases any characteristic 0 field):

%

\begin{table}[h!]\label{tab:table1}
\begin{center}
\def\arraystretch{1.5}
\begin{tabu}{c|[1pt]c|c|c|c|c}
    &   $\sum\limits_{i=1}^n2x_i=1$ &  $\sum\limits_{i=1}^n2^ix_i=-1$  & $\negIm{\sum\limits_{i=1}^nx_i}$   & $\text{PHP}^m_n$ \begin{small}\begin{footnotesize}(CNF)\end{footnotesize}\end{small} &    $\imAx{\sum\limits_{i=1}^nx_i}$ \\ \tabucline[1pt] \\ 
\hline
t-l \reslin{\Q}  &     $2^{\Omega(\sqrt{n})}$  &   $2^{\Omega(n)}$             & $2^{\Omega(n)}$ & $2^{\Omega(n)}$ & $2^{\Omega(\sqrt{n})}$ \\
\hline
t-l \reslinsw{\Q} & \textsf{poly}   & \textsf{poly}                      & $2^{\Omega(n)}$ & $2^{\Omega(n)}$ & \textsf{poly} \\
\hline
\reslin{\Q} & \textsf{poly} & $2^{\Omega(n)}$ & \textsf{poly} &  \textsf{poly} \scriptsize \cite{RT07} & \textsf{poly}\\

\end{tabu}
\end{center}
\caption{\small Lower and upper bounds for \Q. The notation  t-l \reslin{\text{\R}} stands for tree-like \reslin{\R}. The rightmost column
describes bounds on \textit{derivations}, in contrast to refutations. All results except the upper bound on PHP are from the current  work.}
\label{tab:1}
\end{table}

\subsubsection{Finite Fields Lower Bounds} 

We now turn to resolution over linear equations in \emph{finite fields.} We obtain many new tree-like lower bounds (see Table \ref{tab:2}). 



We  already discussed above  lower bounds for the pigeonhole principle which hold both for positive and zero characteristic. We furthermore prove a separation between tree-like \reslin{\text{$\F_{p^k}$}} (resp.~tree-like \reslinsw{\text{$\F_{p^k}$}}) and tree-like \reslin{\text{$\F_{q^l}$}}
(resp.~tree-like \reslinsw{\text{$\F_{q^l}$}})
for every pair of distinct primes $p\neq q$ and every $k,l\in \N\setminus\{0\}$. The separating instances are mod $p$ Tseitin formulas $\text{TS}^{(p)}_{G,\sigma}$ (written as CNFs), which are reformulations of the standard Tseitin graph formulas  $\text{TS}_G$ for counting mod $p$. Furthermore, we establish an exponential lower bound for tree-like \reslinsw{\text{$\F_{p^c}$}} on random $k$-CNFs.\footnote{We thank Dmitry Itsykson for telling us about the lower bound for random $k$-CNF for the case of tree-like \reslin {\F_2}, that was proved by  Garlik and Ko\l odziejczyk using size-width relations (unpublished note). Our result extends Garlik and Ko\l odziejczyk's result to all finite fields. Similar to their result, we   use a size-width argument and simulation by the polynomial calculus to establish the lower bound.}

The lower bounds for tree-like \reslin{\F} for finite fields $\F$ are obtained via a variant of the size-width relation for tree-like
\reslin{\F} together with a translation to polynomial calculus over the field \F, denoted \PCa{\F} \cite{CEI96}, such that \reslin{\F}
proofs of width $\omega$ are translated to \PCa{\F} proofs of degree $\omega$ (the \emph{width} $\omega$ of a clause is defined to be
the total number of disjuncts in a clause). This establishes the lower bounds for the size of tree-like \reslin{\F} proofs via lower
bounds on \PCa{\F} degrees. 

We show that 
$$
\omega_0(\phi\vdash\perp)=O\left(\omega_0(\phi)+\log{S_{\text{t-l\ \reslin{R}}}(\phi\vdash\perp)}\right),
$$
where $\omega_0$ is what we call the \emph{principal width}, which counts the number of linear equations in clauses when we treat as identical
those defining parallel hyperplanes, and $S_{\text{t-l\ \reslin{R}}}(\phi\vdash\perp)$ denotes the minimal size of a tree-like \reslin{R} 
refutation of $\phi$. 

Specifically, over finite fields the following upper and lower bounds provide exponential separations:

\begin{theorem*}[Theorem~\ref{sizeWidth}; Size-width relation]
Let $\phi$ be an unsatisfiable set of linear clauses over a field \F. The following relation between principal width and size holds for both  tree-like \reslin{\F} and
 \treslinsw{\F}:
${S(\phi\vdash\perp)=2^{\Omega(\omega_0(\phi\vdash\perp)-\omega_0(\phi))}}$. If $\F$ is a finite field, then the same relation holds for 
the (standard) width of a clause $\omega$. 
\end{theorem*}

This extends to every field a result by Garlik-Ko\l odziejczyk \cite[Theorem 14]{GK17} who showed a size-width
relation for a system denoted tree-like $\text{PK}^{\text{id}}_{O(1)}(\oplus)$, which is a system extending tree-like
\reslin{\F_2} by allowing arbitrary constant-depth De Morgan formulas as inputs to $\oplus$ (XOR gates)
(though note that our result does not deal with \emph{arbitrary }constant-depth formulas).

\begin{theorem*}[Theorem~\ref{pcsim}]
Let $\F$ be a field and $\pi$ be a \reslin{\F} refutation of an unsatisfiable CNF formula  $\phi$. Then, there exists a \PCa{\F} 
refutation $\pi'$ of (the arithmetization of) $\phi$ of degree $\omega(\pi)$.
\end{theorem*}

\begin{corollary*}[Corollary~\ref{tsLB}; Tseitin mod $p$  lower bounds]
For any fixed prime $p$ there exists a constant $d_0=d_0(p)$ such that the following holds. If $d\geq d_0$, $G$ is a $d$-regular directed graph
satisfying certain expansion properties, and $\F$ is a finite field such that $char(\F)\neq p$, then every tree-like \reslin{\F} refutation
of the Tseitin mod $p$ formula  $\neg\text{\normalfont TS}^{(p)}_{G,\sigma}$ has size $2^{\Omega(dn)}$.
\end{corollary*}

\begin{corollary*}[Corollary~\ref{rndLB}; Random $k$-CNF formulas lower bounds]
Let $\phi$ be a randomly generated $k$-CNF with clause-variable ratio $\Delta$, and where  $\Delta=\Delta(n)$ is such that
$\Delta=o\left(n^{\frac{k-2}{2}}\right)$, and let \F\ be a finite field. Then, every tree-like \reslin{\F} refutation of
$\phi$ has size $2^{\Omega\left(\frac{n}{\Delta^{2/(k-2)}\cdot\log{\Delta}}\right)}$ with probability $1-o(1)$.
\end{corollary*}

\begin{remark}
We  stress that the size-width relation of Theorem~\ref{sizeWidth} \textbf{cannot} be used for transferring \PCa{\F} degree lower bounds
to \treslin{\F} size lower bounds in case $char(\F)=0$. This is due to the essential difference between principal width and width in this case. Thus,
all the lower bounds that we prove using Prover-Delayer games techniques in case $char(\F)=0$ \textbf{do not} follow from lower bounds for \PCa{\F}.
\end{remark}

\smallskip 

Table \ref{tab:2} shows the results for \reslin{\R} over finite fields. 

\begin{table}[h!]
\begin{center}\def\arraystretch{1.5}
\begin{tabu}{l|[1pt]c|c|c|c|c}
                           &   $A\overline{x}=\overline{b}$   & $\text{TS}^{(-)}_{G,\sigma}$ & $\text{TS}^{(q)}_{G,\sigma}$  & random $k$-CNF &    $\text{PHP}^m_n$ \\
\tabucline[1pt] \\ 
t-l \reslin{\text{$\F_{p^k}$}} &   $2^{\Omega(n)}$                             & \textsf{poly} & $2^{\Omega(dn)}$ &
 $2^{\Omega\left(\frac{n}{\Delta^{2/(k-2)}\cdot\log{\Delta}}\right)}$ & $2^{\Omega(n)}$ \\
\hline
t-l Res($\oplus$)  & \textsf{poly} \mbox{\scriptsize \cite{IS14}} & \textsf{poly} \mbox{\scriptsize \cite{IS14}} & $2^{\Omega(dn)}$   & $2^{\Omega\left(\frac{n}{\Delta^{2/(k-2)}\cdot\log{\Delta}}\right)}$
\mbox{\scriptsize \cite{GK17}}  & $2^{\Omega(n)}$ {\scriptsize \cite{IS14}}\\ 

\hline
t-l \reslinsw{\text{$\F_{p^k}$}} & \textsf{poly}                         & \textsf{poly} & \circled{?} & \circled{?} & $2^{\Omega(n)}$\\
\end{tabu}
\end{center}
\caption{\small Lower bounds over finite fields. Here $G$ is $d$-regular graph and $\Delta$ is the clause density (number of clauses divided by the number of variables), $A\overline{x}=\overline{b}$ stands for a linear system over  $\F_{p^k}$ that has no 0-1 solutions in the first and the third rows,  and in the second row the linear system $A\overline{x}=\overline{b}$ is over $\F_2$. The notation $\text{TS}^{(-)}_{G,\sigma}$ stands for
 $\text{TS}^{(p)}_{G,\sigma}$ in the first and the third rows and for
$\text{TS}^{(2)}_{G,\sigma}$ in the second row.
 t-l \reslin{\text{\R}} stands for tree-like \reslin{\R}, and $p\neq q$ are primes (in the second row and third column we assume $q\neq 2$).
 Circled ``?'' denotes an open problem. The results marked with \cite{IS14,GK17} were proved in the respective papers. All other results are from the current work.}
\label{tab:2}
\end{table}


\subsubsection{Complexity of Linear Systems} 

The tree-like \reslin{\F} upper bounds for  mod $p$ Tseitin formulas in the case $char(\F)=p$ stem from  the following proposition:
\begin{proposition*}[Proposition~\ref{linSysUB}; Upper bounds on unsatisfiable linear systems] Let \F\ be a field and assume that the linear system $A\, \overline x=\overline b$, where $A$ is a $k\times n$ matrix over \F, has no solutions (over \F).
Let $\phi$ be a CNF formula encoding  the linear system $A\, \overline x=\overline b$. Then, there exist tree-like \reslin{\F} refutations of $\phi$ of size
polynomial in the sum of sizes of encodings of all coefficients in $A$.
\end{proposition*}

The upper bound in Proposition~\ref{linSysUB} applies only to linear systems that are unsatisfiable over the \emph{\uline{whole}} field \F. 
But does any system $A\, \overline x=\overline b$ over \F\ that has a satisfying assignment over \F, but \emph{not} over 0-1 assignments,
admit  polynomial-size \reslin{\F} refutations? 

For fields $\F$ with $char(\F)\geq 5$ or $char(\F)=0$ it is known that 0-1 satisfiability of $A\, \overline x=\overline b$ is \NP-complete (see Sec.~\ref{sec:linSysCompl}). This means that unless $\coNP=\NP$ there exist 0-1 unsatisfiable linear systems that require superpolynomial dag-like \reslin{\F} refutations. Moreover, the reduction $R$ from $k$-UNSAT is such that $\phi\in k\text{-UNSAT}$ 
has \reslin\F\ refutations of size $S$ iff the system $R(\phi)$ has \reslin\F\ refutations of size $O(S)$. Thus, in general proving lower bounds for linear systems can be as hard as proving lower bounds for CNFs: lower bounds for some linear systems imply
lower bounds for CNFs.  

An unconditional explicit bound for \treslin{\F} can be obtained via \PCa{\F} using size-width relation for finite fields (Theorem~\ref{sizeWidth}) and Proposition~\ref{thm:linSysPC}. In particular, hard instances of the form $A\, \overline x=\overline b$ can be constructed by applying the reduction in the proof of \NP-completeness of 0-1 satisfiability of linear systems to, say, mod 2 Tseitin formulas. Our work implies an exponential lower bound for the size of \treslin{\F} refutations of these systems (for large enough, but constant, characteristic) and we conjecture that they are hard for dag-like \reslin\F\ as well.    

We prove an upper bound for linear systems and suggest another, more direct, construction of a hard candidate, using error-correcting codes.

\begin{theorem*}[Theorem~\ref{thm:linSysImUB}; Upper bound on 0-1 unsatisfiable linear systems]
Let $A_{f_1,\ldots,f_m}:\F^n\rightarrow \F^m$ be an affine map $\overline{x}\mapsto (f_1(\overline{x}),\ldots, f_m(\overline{x}))$, where $f_1,\ldots,f_m$ are linear forms.
If the system $f_1=0,\ldots,f_m=0$ is unsatisfiable over 0-1, that is, if $0\notin im_2(A_{f_1,\ldots,f_m}\, \overline{x})$, then there exists a \reslin{\F} refutation of this
system of size $poly(n+|im_2(A_{f_1,\ldots,f_m}\, \overline{x})|)$.
\end{theorem*}

The instance is constructed specifically to be  hard for a simple and natural model of decision trees, which can be simulated both by  tree-like \reslin{\F} and \PCa{\F}
and reflects a natural strategy to refute 0-1 unsatisfiable linear systems. Such a strategy for refuting $A\, \overline{x}=\overline{b}$ can be informally described as follows:
select variables and try to assign them 0-1 values until the system $(A\, \overline{x}=\overline{b})\rst_{\rho}$ becomes unsatisfiable over \F, where $\rho$ is the
current assignment, and refute it by a polynomial-size refutation, guaranteed by Proposition~\ref{linSysUB} (above). Formally, a decision tree for $A\, \overline{x}=\overline{b}$
is a binary decision tree, where every leaf is marked with unsatisfiable over \F\ system $(A\, \overline{x}=\overline{b})\rst_{\rho}$, where $\rho$ consists of variable
assignments on the path from the root to the leaf.

The matrix $A$ of the instance is constructed as a generator matrix of a linear error-correcting $(n,k,d)_q$ code, where $n$ is the code length, $k$ is the dimension of the code space, $d$ is the minimal distance of the code and $q=|\F|$. The parameter $k$ is chosen to be large enough to ensure that $q^k>2^n$ and thus there exists some 
$\overline{b}$ such that $A\, \overline{x}=\overline{b}$ has no 0-1 solutions. On the other hand, $d=\Omega(\frac{n}{\log n})$ is chosen to be large enough to ensure 
that all the leaves of a decision tree for $A\, \overline{x}=\overline{b}$ are sufficiently deep in the tree: if $\rho$ assigns at most $k<d$ variables, then the code generated by 
$A\rst_{\rho}$ has a minimal distance at least $d-k$ and therefore $A\rst_{\rho}$ has full rank. The existence of this code is guaranteed by the Gilbert-Varshamov bound.

\begin{theorem*}[Theorem~\ref{thm:linSysDTLB}; Lower bound for decision trees on linear systems]
For every $n\in\N$ there exists a 0-1 unsatisfiable linear system $A\, \overline{x}=\overline{b}$ over a finite field $\F_q$, $q > 2$, with $n$ variables,
such that any decision tree for this system is of size $2^{\Omega\left(\frac{n}{\log n}\right)}$.
\end{theorem*}

\subsubsection{Nondeterministic Linear Decision Trees} 
There is a well-known size preserving (up to a constant factor) correspondence between tree-like resolution refutations for
unsatisfiable formulas $\phi$ and decision trees, which solve the following problem: given an assignment $\rho$ for the variables of $\phi$, determine
which clause $C\in\phi$ is falsified by querying values of the variables under the assignment $\rho$. In Itsykson-Sokolov \cite{IS14} this correspondence
was generalized to  tree-like $\text{Res}(\oplus)$ refutations and parity decision trees. In the current work we initiate the study of linear decision trees and their properties over different characteristics, extending the correspondence to
a correspondence between \treslin{R} (and \treslinsw{R}) derivations to what we call  \emph{nondeterministic linear decision trees} (NLDT). 

NLDTs for
an unsatisfiable set of linear clauses $\phi$ are binary rooted trees,
where every edge is labeled with a non-equality $f\neq 0$ for a linear form $f$ and every leaf is labeled with a linear clause $C\in\phi$,
which is violated by the non-equalities on the path from the root to the leaf. (Note that in the same manner that in a (boolean) decision tree (which corresponds to a tree-like resolution refutation) we go along a path from the root to a leaf, choosing those edges that violate a literal $x_i$ or $\neg x_i$, in an NLDT we branch   along a path that violates equalities $f=0$, or equivalently, certifies non-equalities of the form $f\neq 0$.)  

\begin{theorem*}[Theorem~\ref{treeDTequiv}]
If $\phi$ is an unsatisfiable CNF formula, then every tree-like \reslin{R} or tree-like \reslinsw{R} refutation can be transformed into a corresponding  NLDT for $\phi$ of the same size up to a constant factor, and vice versa (note that the NLDTs for the two types of refutations are different).
\end{theorem*}

\section{Preliminaries}

\subsection{Notation}\label{sec:notations}

Denote by  $[n]$ the set $\{1,\dots,n\}$.
We use $x_1,x_2,\ldots$ to denote variables, both propositional and  algebraic. Let $f$ be a linear polynomial (equivalently, an affine function) over a ring \R, that is, a function of the form $\sum_{i=1}^n a_ix_i+a_0$ with $a_i\in\R$. We sometimes refer to a linear form as a \emph{hyperplane}, since a linear form determines a hyperplane. We denote by $im_2(f)$ the image of $f$ under 0-1 assignments to its variables; ${\langle f\neq A\rangle}:={\bigvee\nolimits_{A\neq B\in im_2(f)}(f=B)}$, where $A\in \R$. 

A \emph{linear clause} is a formula of the form $\left(\sum\nolimits_{i=1}^na_{1i}x_i+b_1=0\right)\vee\dots\vee\left(\sum\nolimits_{i=1}^na_{ki}x_i+b_k=0\right)$ with $x_1,\dots,x_n$ variables, and $a_{ij},b_i$'s ring elements (when the ring is specified in advanced). We sometimes abuse notation by writing a linear equation as $\sum\nolimits_{i=0}^na_{1i}x_i=-b_1$ instead of $\sum\nolimits_{i=0}^na_{1i}x_i+b_1=0$.
We assume that all the disjuncts in a linear clause are distinct. 

For  $\phi$  a set of clauses or linear clauses,  $vars(\phi)$ denotes the set of variables occurring in $\phi$ and let $\text{Vars}$ denote the set of \emph{all} variables.

Let $A$ be a matrix over a ring. We introduce the notation $Ax\doteqdot b$ for a system of linear non-equalities, where a \bemph{non-equality} means $\neq$ (note the difference between $Ax\doteqdot b$, which stands for $A_i\cd x\neq b_i$, for \textit{all} rows $A_i$ in $A$, and $Ax\neq b$, which stands for $A_i\cd x\neq b_i$, for \textit{some} row $A_i$ in $A$).

If $f$ is a linear polynomial over \R\ and $A$ is a matrix over \R, denote by $|f|$ the sum of sizes of encodings of coefficients in $f$ and by $|A|$ the sum of sizes of encodings of elements in $A$.

If $C=(\bigvee\nolimits_{i\in[m]}f_i=0)$ is a linear clause, denote by $\neg C$ the  \emph{set} of non-equalities $\{f_i\neq 0\}_{i\in [m]}$. Conversely, 
if $\Phi=\{f_i\neq 0\}_{i\in[n]}$ is a set of
non-equalities, denote $\neg\Phi:=\bigvee\nolimits_{i\in[m]}f_i=0$.

If $\phi$ is a set of linear clauses over a ring $\R$ and $D$ is a linear clause over $\R$, denote by $\bigwedge\nolimits_{C\in\phi}C\models D$ and $\bigwedge\nolimits_{C\in\phi}C\models_{\R} D$
semantic entailment over 0-1 and $\R$-valued assignments respectively. 

Let $l$ be a linear polynomial not containing the variable $x$. If $C$ is a linear clause, denote by $C\!\rst_{x\leftarrow l}$ the linear clause, which is obtained
from $C$ by substituting $l$ for $x$ everywhere in $C$. If $\phi=\{C_i\}_{i\in I}$ is a set of clauses, denote $\phi\rst_{x\leftarrow l}:=\{C_i\rst_{x\leftarrow l}\}_{i\in I}$.
We define a \emph{linear substitution} $\rho$ to be a sequence $(x_1\leftarrow l_1,\ldots,x_n\leftarrow l_n)$ such that each linear polynomial $l_i$ does not depend on $x_i$. For a clause or
a set of clauses $\phi$ we define $\phi\rst_{\rho}:=(\ldots((\phi\rst_{x_1\leftarrow l_1})\rst_{x_2\leftarrow l_2})\ldots)\rst_{x_n\leftarrow l_n}$.


\subsection{Propositional Proof Systems}\label{sec:Propositional-Proof-Systems}
A \emph{clause} is an expression of the form $l_1\vee \dots \vee l_k$, where $l_i$ is a literal, where a \emph{literal} is a propositional variable $x$ or its negation
$\neg x$. A formula is in \emph{Conjunctive Normal Form} (CNF) if it is a conjunction of clauses. A CNF can thus be defined simply as a set of clauses. The 
choice of a reasonable binary encoding of sets of clauses allows us to define the language $\text{UNSAT}\subset\{0,1\}^*$ of unsatisfiable propositional formulas in CNF.
We sometimes interpret an element in $\text{UNSAT}$ as a formula and sometimes as a set of clauses. 
Dually, a formula is in \emph{Disjunctive Normal Form} (DNF) if it is a disjunction of conjunctions of literals and TAUT is the language of tautological
propositional formulas in DNF. There is a bijection between TAUT and UNSAT, which preserves the size of the formula, given by negation. 

A formula is in $k$-CNF (resp.~$k$-DNF) if it is in CNF (resp.~DNF) and every clause (resp.~conjunct) has at most $k$ literals. $k$-UNSAT (resp.~$k$-TAUT) is the
language of unsatisfiable (resp.~tautological) formulas in $k$-CNF (resp.~$k$-DNF).

\begin{definition}[Cook-Reckhow propositional proof system \cite{CR79}]\label{PPSdef}
A \emph{propositional proof system} $\Pi$ is a polynomial time computable onto function $\Pi:\{0,1\}^*\rightarrow \text{TAUT}$.
\end{definition}
$\Pi$-proofs of $\phi\in\text{TAUT}$ are elements in $\Pi^{-1}(\phi)$. Definition~\ref{PPSdef} can be generalized
to arbitrary languages: proof system for a language $L$ is polynomial time computable onto function $\Pi:\{0,1\}^*\rightarrow L$. In particular,  
a \emph{refutation system} $\Pi$ is a proof system for UNSAT. Post-composition with negation turns a propositional proof system  into a refutation system and vise versa.

Denote by $S(\pi)$, and alternatively by $|\pi|$, the size of the binary encoding of a proof $\pi$ in a proof system $\Pi$. For 
$\phi\in\text{UNSAT}$ and a refutation system $\Pi$ denote by $S_{\Pi}(\phi\vdash \perp)$ (we sometimes omit the subscript $\Pi$ when it is clear from the context) the minimal
size of a $\Pi$-refutation of $\phi$.

The \emph{resolution} system (which we denote also by  Res) is a refutation system, based on the following rule, allowing to derive new clauses from  given ones:

\begin{prooftree}
        \centering
        \def\labelSpacing{12pt}
        \AxiomC{$C \vee x$}
        \AxiomC{$D \vee \neg x$}
        \RightLabel{(Resolution rule).}
        \BinaryInfC{$C \vee D$}
\end{prooftree}        
A \emph{resolution derivation} of a clause $D$ from a set of clauses $\phi$ is a sequence of clauses $(D_1,\dots,D_s\equiv D)$ such that for every $1\leq i\leq s$ 
either $D_i\in \phi$ or $D_i$ is obtained from previous clauses by applying the resolution rule. A \emph{resolution refutation} of $\phi\in\text{UNSAT}$ is a 
resolution derivation of the empty clause from $\phi$, which stands for the truth value \textsf{False}. 

A resolution derivation is \emph{tree-like} if every clause in it is used at most once as a premise of a rule. Accordingly, \emph{tree-like resolution} is the resolution system allowing only tree-like refutations.

Let $\F$ be a field. A \emph{polynomial calculus} \cite{CEI96} derivation of a polynomial $q\in \F[x_1,\dots,x_n]$ from a set of polynomials
$\mathcal{P}\subseteq \F[x_1,\dots,x_n]$ is a sequence $(p_1,\dots,p_s),p_i\in \F[x_1,\dots,x_n]$ such that for every $1\leq i\leq s$ either $p_i=x_j^2-x_j$,
$p_i\in\mathcal{P}$ or $p_i$ is obtained from previous polynomials by applying one of the following rules:\vspace{-14pt} 
\begin{prooftree}
        \centering
        \def\labelSpacing{12pt}
        \AxiomC{$f$}
        \AxiomC{$g$}
        \RightLabel{($\alpha,\beta\in \F,f,g\in \F[x_1,\dots,x_n]$)}
        \BinaryInfC{$\alpha f + \beta g$}

        \AxiomC{$f$}
        \RightLabel{($f\in \F[x_1,\dots,x_n]$)\,.}
        \UnaryInfC{$x\cdot f$}
        \noLine
        \BinaryInfC{}

\end{prooftree}
A polynomial calculus refutation of $\mathcal{P}\subseteq \F[x_1,\dots,x_n]$ is a derivation of $1$. The degree $d(\pi)$ of a polynomial calculus derivation $\pi$ is the maximal total degree of
a polynomial appearing in it.
This defines the proof system \PCa{\F} for the language of unsatisfiable systems of polynomial equations over $\F$. It can be turned into a proof system for $k$-UNSAT via
\emph{arithmetization of clauses} as follows: $(x_1\vee\ldots\vee x_k\vee \neg y_1\vee\ldots\vee \neg y_l)$ is represented as $(1-x_1)\cdot\ldots\cdot(1-x_k)\cdot y_1\cdot\ldots\cdot y_l=0$.

\subsection{Hard Instances}
\subsubsection{Pigeonhole Principle}\label{sec:PHP-intro}

The \emph{pigeonhole principle} states that there is no injective mapping from the set $[m]$  
to the set $[n]$, for $m>n$. Elements of the former and the latter sets are  referred to as \emph{pigeons} and \emph{holes}, respectively.
The CNF formula, denoted $\text{PHP}^m_n$, encoding the negation of this principle is defined as follows. Let the set of propositional variables
$\{x_{i,j}\}_{i\in [m],j\in [n]}$ correspond to the mapping from $[m]$ to $[n]$, that is, $x_{i,j}=1$ iff the $i^{\text{th}}$ pigeon is mapped 
to the $j^{\text{th}}$ hole. Then 
$\neg\text{PHP}^m_n:=\phpInj^m_n\cup \phpTot^m_n\in\text{UNSAT}$, where 
$\phpInj^m_n=\{\bigvee\nolimits_{j\in [n]}x_{i,j}\}_{i\in [m]}$ are axioms for pigeons and 
$\phpTot^m_n=\{\neg x_{i,j} \vee \neg x_{i',j}\}_{i\neq i'\in [m],j\in[n]}$ are axioms for holes.

Weaker (namely, easier to refute) versions of $\neg\text{PHP}^m_n$ are obtained by augmenting it with the \emph{functionality} axioms 
$\text{Func}^m_n:= \{\neg x_{i,j} \vee \neg x_{i,j'}\}_{i\in [m],j\neq j'\in [n]}$ ($\neg\text{FPHP}^m_n$) or the \emph{surjectivity}
axioms $\text{Surj}^m_n:=\{\bigvee\nolimits_{i\in [m]}x_{i,j}\}_{j\in [n]}$ ($\neg \text{onto-PHP}^m_n$).

\subsubsection{Mod $p$ Tseitin Formulas}\label{sec:tsDef}

We use the version given in \cite{AR01} (which is different from the one in \cite{BGIP01,RT07}). Let $G=(V,E)$ be a directed $d$-regular graph. We 
assign to every edge $(u,v)\in E$ a corresponding variable $x_{(u,v)}$. Let $\sigma:V\rightarrow \F_p$. The \emph{Tseitin mod $p$ formulas}
$\neg\text{TS}^{(p)}_{G,\sigma}$ are the CNF encoding of the following equations for all $u\in V$:

\begin{equation}\label{eq:tseitin-mod-p}
\sum\limits_{(u,v)\in E}x_{(u,v)}-\sum\limits_{(v,u)\in E}x_{(v,u)}\equiv\sigma(u)\mod{p}\,.
\end{equation}
Note that we use the standard encoding of boolean functions as CNF formulas and the number of clauses, required to encode these equations is $O(2^d|V|)$. $\neg\text{TS}^{(p)}_{G,\sigma}$ is unsatisfiable if $\sum\nolimits_{u\in V}\sigma(u)\not \equiv 0\mod{p}$. To see this, note that if we sum \eqref{eq:tseitin-mod-p} over all nodes $u\in V$ we obtain precisely $\sum\nolimits_{u\in V}\sigma(u)$ which is different from $0\mod{p}$; but on the other hand, in this sum over all nodes $u\in V$ each edge $(u,v)\in E$ appears once with a positive sign as an outgoing edge from $u$ and with a negative sign as an incoming edge to $v$, meaning the the total sum is 0, which is a contradiction.  

In particular, $\neg\text{TS}^{(2)}_{G,\sigma}$ are the classical Tseitin
formulas \cite{Tse68} and $\text{TS}^{(2)}_{G,1}$, where $1$ is the constant function $v\mapsto 1$ (for all $v\in V$), expresses the fact that the sum of total degrees (incoming $+$
outgoing) of the vertices is even.

The proof complexity of Tseitin tautologies depends on the properties of the graph $G$. For example, if $G$ is just a union of $K_{d+1}$ (the complete graphs
on $d+1$ vertices), then they are easy to prove. On the other hand, they are known to be hard for some proof systems if $G$ satisfies certain expansion properties. 

Let $G=(V,E)$ be an \emph{undirected} graph. For $U,U'\subseteq V$ define $e(U,U'):=\{(u,u')\in E\,|\,u\in U,u'\in U'\}$. Consider the following
measure of expansion for $r\geq 1$:
$$c_E(r,G):=\min\limits_{|U|\leq r}\frac{e(U,V\backslash U)}{|U|}$$
$G$ is $(r,d,c)$-expander if $G$ is $d$-regular and $c_E(r,G)\geq c$. There are explicit constructions of good expanders. For example:

\begin{proposition}[Lubotzky \textit{et.~al}~\cite{LPS88}]
For any $d$, there exists an explicit construction of $d$-regular graph $G$, called Ramanujan graph, which is 
$(r,d,d(1-\frac{r}{n})-2\sqrt{d-1})$-expander for any $r\geq 1$.
\end{proposition}

\begin{proposition}[Alekhnovich-Razborov \cite{AR01}]
For any fixed prime $p$ there exists a constant $d_0=d_0(p)$ such that the following holds. If $d\geq d_0$, $G$ is a $d$-regular Ramanujan
graph on $n$ vertices (augmented with arbitrary orientation of its edges) and $char(\F)\neq p$, then for every function $\sigma$ such that
$\neg\text{TS}^{(p)}_{G,\sigma}\in \text{UNSAT}$ every \PCa{\F} refutation of $\neg\text{TS}^{(p)}_{G,\sigma}$ has degree $\Omega(dn)$.
\end{proposition}

\subsubsection{Random k-CNFs}

A random $k$-CNF is a formula $\phi\sim\mathcal{F}^{n,\Delta}_k$ with $n$ variables that is generated by picking randomly and independently
$\Delta\cdot n$ clauses from the set of all ${n \choose k}\cdot 2^k$ clauses.

\begin{proposition} [Alekhnovich-Razborov \cite{AR01}]
Let $\phi\sim\mathcal{F}^{n,\Delta}_k,k\geq 3$ and $\Delta=\Delta(n)$ is such that $\Delta=o\left(n^{\frac{k-2}{2}}\right)$. Then every \PCa{\F} refutation of
$\phi$ has degree $\Omega\left(\frac{n}{\Delta^{2/(k-2)}\cdot\log{\Delta}}\right)$ with probability $1-o(1)$ for any field $\F$.
\end{proposition} 

\subsection{Error-Correcting Codes}

\begin{definition}[\cite{Betten06}]\label{def:ECC}
Let $A:\F_q^k\hookrightarrow\F_q^n$ be a linear embedding. The image $C=im(A)$ of $A$ is called $(n,k,d)_q$-code if for any 
$\overline{x},\overline{y}\in C$ it holds that  $d_H(\overline{x},\overline{y})\geq d$, where 
$d_H(\overline{x},\overline{y})=|\{i\, |\, x_i\neq y_i\}|$ is the Hamming distance. The matrix of $A$ is called generator matrix for $C$.
\end{definition}

\begin{theorem}[Gilbert bound \cite{Betten06}]\label{thm:Gilbert}
If $q$ is a power of a prime and $n,k,d\in\N, n\geq k$ are such that inequality
$$\sum\limits_{i=1}^d\left(\begin{array}{c}n \\ i\end{array}\right)\cd (q-1)^i<q^{n-k+1}$$
holds, then there exists $(n,k,d)_q$-code.
\end{theorem}

\subsection{Complexity of Linear Systems}\label{sec:linSysCompl}

It is a well-known fact that deciding 0-1 satisfiability of linear systems over $\F_p,p\geq 5$ or of linear systems over
$\Q$ (even if coefficients are small) are \NP-complete problems. Indeed, for example, the $3$-clause $(x_1\vee \neg x_2\vee x_3)$ can  
be represented as the linear equation with additional boolean  variables $y_1, y_2$: $x_1+(1-x_2)+x_3=1+y_1+y_2$. In this
way $k$-SAT reduces to 0-1 satisfiability of linear systems over a field of characteristic $0$ or $p>k$.

\begin{theorem}\label{thm:linSysNP}
The problem of deciding 0-1 satisfiability of linear systems over a field of characteristic $0$ or $p\geq 5$ is \NP-complete.
In case of characteristic $0$ this also holds if the size of coefficients is required to be bounded by a constant.
\end{theorem}

The mapping $R$ of $k$-CNFs to linear systems described above can be used to translate lower bounds on degree of \PCa{\F} refutations
from $k$-CNFs to linear systems. 

\begin{proposition}\label{thm:linSysPC}
If $\phi\in k\text{-UNSAT}$ and \F\ is a field such that $char(\F)>k$ or $char(\F)=0$, then $\phi$ admits \PCa{\F} refutations of
degree $d$ iff $R(\phi)$ admits \PCa{\F} refutations of degree $O(d)$. 
\end{proposition}
\begin{proof}
Denote $\sigma$ the mapping from literals to linear polynomials such that: $\sigma(x):=x$ and $\sigma(\neg x):= 1-x$. Let
$\tau$ be the following mapping from clauses to linear polynomials: 
$\tau(l_1\vee\dots\vee l_s):=\sigma(l_1)+\dots+\sigma(l_s)-1-y^{(1)}_{l_1\vee\dots\vee l_s}-\dots-y^{(s-1)}_{l_1\vee\dots\vee l_s}$, where
$y^{(i)}_{l_1\vee\dots\vee l_s}$ are auxiliary boolean  variables. Then $R$ translates $\phi=\{C_i\}_{i\in[m]}$ to the 0-1 unsatisfiable linear system $L$:
$\tau(C_1)=0,\ldots,\tau(C_m)=0$. 

Assume $L$ has \PCa{\F} refutation $\pi$ of degree $d$. If $x_1,\ldots,x_n$ are variables of $\phi$, then all the auxiliary 
variables $y^{(i)}_{C_j}$ can
be substituted with polynomials $v^{(i)}_{C_j}(x_1,\ldots,x_n)$ of degree at most $k$ such that $C_j\models (\tau(C_j)\rst_{\rho_v})=0$,
where $\rho_v$ stands for the substitution and the entailment is over 0-1 assignments. It is easy to see that $\pi$ can be extended 
to the proof $\pi\rst_{\rho_v}$ of degree at most $k\cd d$, where all the auxiliary variables are substituted with the corresponding polynomials.
Due to implicational completeness of \PCa{\F}, there are \PCa{\F} derivations $\pi_j : C_j\vdash (\tau(C_j)\rst_{\rho_v})=0$ of degree
at most $k$. Composition of $\{\pi_j\}_{j\in[m]}$ with $\pi\rst_{\rho_v}$ gives a \PCa{\F} refutation of degree at most $k\cd d$.

Conversely, if $\pi$ is a \PCa{\F} refutation of $\phi$ of degree $d$, then the composition of derivations $\tau(C_j)=0\vdash C_j$
with $\pi$ gives a refutation of $L$ of degree at most $\max(k,d)$.
  
\end{proof} 

%
\section{Resolution over Linear Equations for General Rings}\label{sec:resLin}

In this section we define and outline some basic properties of systems that are extensions of resolution, where clauses are disjunctions of linear equations over a ring $R$: 
$\left(\sum\nolimits_{i=0}^na_{1i}x_i+b_1=0\right)\vee\dots\vee\left(\sum\nolimits_{i=0}^na_{ki}x_i+b_k=0\right)$. Recall that disjunctions of this form are called \emph{linear clauses}, and that we assume that all disjuncts are distinct, hence contract duplicate linear equations.
We sometimes abuse notation by writing a linear equation as $\left(\sum\nolimits_{i=0}^na_{1i}x_i=-b_1\right)$ instead of $\left(\sum\nolimits_{i=0}^na_{1i}x_i+b_1=0\right)$.

The rules of \reslin{R} are as follows (cf.~\cite{RT07}): 

\begin{prooftree}
        \centering
        \def\labelSpacing{12pt}
        \AxiomC{$C \vee f(\overline{x})=0$}
        \AxiomC{$D \vee g(\overline{x})=0$}
        \RightLabel{($\alpha,\beta\in R$)}
        \LeftLabel{(Resolution)}
        \BinaryInfC{$C \vee D \vee \left(\alpha f(\overline{x})+ \beta g(\overline{x})\right)=0$}

\end{prooftree}

\begin{prooftree}
        \AxiomC{$C \vee a=0$}
        \RightLabel{($0\neq a\in R$)}
        \LeftLabel{(Simplification)}
        \UnaryInfC{$C$}
        
        \AxiomC{$C$}
        \LeftLabel{(Weakening)}
        \UnaryInfC{$C \vee f(\overline{x})=0$}
        \noLine
        \BinaryInfC{}
\end{prooftree}
where $f(\overline{x}),g(\overline{x})$ are linear forms over $R$ and $C,D$ are linear clauses. Note that contraction of duplicates disjuncts is done automatically when applying the resolution rule.  The \emph{boolean  axioms} are defined as follows: 
$$ x_i=0 \lor x_i=1 \text{, ~for $x_i$ a variable}$$
A \reslin{R} \emph{derivation} of a linear clause
$D$ from a set of linear clauses $\phi$ is a
sequence of linear clauses $(D_1,\dots,D_s\equiv D)$ such that for every $1\leq i\leq s$ 
either $D_i\in \phi$ or is a boolean  axiom or $D_i$ is obtained from previous clauses by applying one of the rules above.
A \reslin{R}
\emph{refutation} of an unsatisfiable
set of linear clauses $\phi$ is a \reslin{R} derivation of the empty clause (which stands for \textsf{false}) from $\phi$. The \bemph{size} of a \reslin\R\ derivation  is the total size of all the clauses in the derivation, where the size of a clause is defined to be the total number of occurrences of variables in it plus the total size of all the coefficient occurring in the clause. The size of a coefficient when using integers (or integers embedded in characteristic zero rings) will be the standard size of the binary representation of integers.



In this definition we assume that $R$ is a non-trivial ($R\neq\textbf{0}$) ring such that there are polynomial-time algorithms for addition, multiplication
and taking additive inverses.






Along with  size, we will be dealing with two complexity measures of derivations: \emph{width} and \emph{principal width}.

\begin{definition}\label{omegaDef}
A clause $C=(f_1=0\vee\dots\vee f_m=0)$ has \textbf{width} $\omega(C)=m$ and \textbf{principal width} $\omega_0(C)=\left|\{f_i\}_{i\in[m]}/_\sim\right|$ 
where $\sim$ identifies \R-linear forms $f_i=0$ and $f_j=0$ if they define parallel hyperplanes, that is, if $f_i=Af_j+B$ or $f_j=Af_i+B$ for some $A,B\in\R$.
For $\mu\in\set{\omega,\omega_0}$, the measure $\mu$ associated with a 
\reslin{R} derivation $\pi=(D_1,\dots,D_s)$ is 
$\mu(\pi):=\max\nolimits_{1\leq i\leq s}\mu(D_i)$. For $\phi\in\text{UNSAT}$, denote by $\mu(\phi\vdash \perp)$ the minimal value of $\mu(\pi)$ over all \reslin{R} refutations $\pi$.
\end{definition}

\begin{proposition}
\reslin{R} is sound and complete. It is also implicationally complete, that is if $\phi$ is a set
of linear clauses and $C$ is a linear clause such that $\phi\models C$, then there exists a \reslin{R} derivation of $C$ from $\phi$.
\end{proposition}
\begin{proof}
The soundness can be checked by inspecting that each rule of \reslin{R} is sound. Implicational completeness (and thus completeness) follows
from Proposition~\ref{univUB}.
\end{proof}

We now define two systems of resolution   with linear equations over a ring, where some of the rules are semantic:
\reslinsw{R} and \reslinsem{R}. \reslinsw{R} is obtained from \reslin{R}
by replacing the boolean axioms with $0=0$, discarding simplification rule and replacing the weakening rule with the following \textit{semantic weakening rule}:


\begin{prooftree}
        \AxiomC{$C$}
        \RightLabel{($C\models D$)}
        \LeftLabel{(Semantic weakening)}
        \UnaryInfC{$D$}
\end{prooftree}

The system \reslinsem{R} has no axioms except for $0=0$, and has only the following \textit{semantic resolution rule}:

\begin{prooftree}
        \AxiomC{$C$}
        \AxiomC{$C'$}
        \RightLabel{($C\wedge C'\models D$)}
        \LeftLabel{(Semantic resolution)}
        \BinaryInfC{$D$}
\end{prooftree}

It is easy to see that $\text{\reslin{R}}\leq_p\text{\reslinsw{R}}\leq_p\text{\reslinsem{R}}$, where $P\le_p Q$ denotes that $Q$ polynomially simulates $P$. 

In contrast to the case $\R=\F_2$
(see \cite{IS14}), for rings \R\ with $char(\R)\notin\{1,2,3\}$ both \reslinsw{R} and \reslinsem{R} are not Cook-Reckhow proof systems,
unless $\P=\NP$: 
\begin{proposition}\label{swHardness}
The following decision problem is \coNP-complete: given a linear clause over a ring R with $char(R)\notin \{1,2,3\}$ decide whether it is a tautology under 0-1 assignments. 
\end{proposition}
\vspace{-12pt} 
\begin{proof}
Consider a 3-DNF $\phi$ and encode every conjunct $(x^{\sigma_1}_{i_1}\wedge \dots \wedge x^{\sigma_k}_{i_k})\in\phi,1\leq k\leq 3,\sigma_i\in\{0,1\}$ as
the equation 
$(1-2\sigma_1)x_1+\dots+(1-2\sigma_k)x_k=k-(\sigma_1+\dots+\sigma_k)$, where $x^0:=x,x^1:=\neg x$. Then $\phi$ is tautological  if and only if the disjunction
of these linear equations is tautological (that is, for every 0-1 assignment to the variables at least one of the equations hold, when the equations are computed
over a ring with characteristic zero or finite characteristic bigger than 3).
\end{proof}


We leave it as an open question to determine the complexity  of verifying a correct application of the semantic weakening in  case $char(\text{\R})=3$ or in case $char(\R)=2$
and $\text{\R}\neq \F_2$. In the case $\R=\F_2$ the negation of a clause is a system of linear equations and thus the existence of
solutions for it can be checked in polynomial time. Therefore \reslinsw{\text{$\F_2$}} is a Cook-Reckhow propositional proof system. The definitions 
of \reslin{\text{$\F_2$}}, \reslinsw{\text{$\F_2$}} and \reslinsem{\text{$\F_2$}} coincide with the definitions of syntactic $\text{Res}(\oplus)$,
$\text{Res}(\oplus)$ and $\text{Res}_{\text{sem}}(\oplus)$ from \cite{IS14}, respectively\footnote{There is, however, one minor difference 
in the formulation of syntactic $\text{Res}(\oplus)$ and \reslin{\text{$\F_2$}}: the former does not have the boolean axioms, but has an extra rule (\emph{addition rule}).}.
As showed in \cite{IS14}, \reslin{\text{$\F_2$}}, \reslinsw{\text{$\F_2$}} and \reslinsem{\text{$\F_2$}} are polynomially equivalent.
\medskip 

We now show that if $char(\R)\notin\{1,2,3\}$, then \reslinsw{R} is polynomially bounded as
a proof system for $3$-UNSAT (that is, admits polynomial-size refutation for every instance): 
\begin{proposition}\label{swPB}
If $char(\R)\notin\{1,2,3\}$, then dag-like \reslinsw{R} and tree-like \reslinsem{R} are polynomially bounded (not necessarily Cook-Reckhow) propositionally proof systems for 3-UNSAT.
\end{proposition}
\begin{proof}
Let $\phi(x_1,\ldots,x_n)=\{C_i\}_{i\in [m]}\in 3\text{-UNSAT}$. Given 
$C=(x^{\sigma_{1}}_{j_{1}}\vee \ldots \vee x^{\sigma_{k}}_{j_{k}})$ define $lin(\neg C):=\left((2\sigma_{1}-1)x_{j_{1}}+\ldots+(2\sigma_{k}-1)x_{j_{k}}-(\sigma_{1}+\ldots+\sigma_{k})\right)$
where $\sigma_{i}\in \{0,1\},j_{l}\in [n],x^0:=x,x^1:=\neg x$. The linear clause
$lin(\neg\phi):=\bigvee\nolimits_{i\in[m]}lin(\neg C_i)=0$ is a  tautology (under 0-1 assignments) and thus can be derived in \reslinsw{R} in a single step as a weakening of $0=0$ or 
resolving $0=0$ with $0=0$ in tree-like \reslinsem{R}.

In tree-like \reslinsem{R} the disjunct $lin(\neg C_i)=0$ can be eliminated from $lin(\neg\phi)$ by a single resolution with $C_i$, thus the empty 
clause is derived by a sequence of $m$ resolutions of $lin(\neg \phi)$ with $C_1,\ldots,C_m$.

Similarly, the disjuncts $lin(\neg C_i)=0$ are eliminated from $lin(\neg\phi)$ in \reslinsw{R}, but with a few more steps.
Let $D_0$ be the empty clause and $D_{s+1}:=D_s\vee lin(\neg C_{s+1})=0,0\leq s < m$. Assume $D_{s+1}$ is
derived and assume without loss of generality, that $C_{s+1}=(x_1=1\vee\ldots\vee x_k=1)$ and thus $lin(\neg C_{s+1})=\left(-x_{1}-\ldots-x_{k}\right)$. Derive 
$D_s$ as follows. Resolve $D_{s+1}$ with $C_{s+1}$ on $lin(\neg C_{s+1})+(x_k-1)$ to get the clause 
$E_1:=D_s\vee \left(-x_{1}-\ldots-x_{k-1}-1\right)=0\vee x_1=1\vee\ldots\vee x_{k-1}=1$ and apply 
semantic weakening to get $E_1':=D_s\vee x_1=1\vee\ldots\vee x_{k-1}=1$. Resolve $D_{s+1}$ 
with $E_1'$ on $lin(\neg C_{s+1})+(x_{k-1}-1)$ and apply semantic weakening to get the clause $E_2':=D_s\vee x_1=1\vee\ldots\vee x_{k-2}=1$.
After $k$ steps the clause $D_s=E_k'$ can be  derived. 
\end{proof}

The following proposition is straightforward, but useful as it allows, for example, to transfer results about \reslin{\Q} to \reslin{\Z}.

\begin{proposition}\label{fracFields}
If $R$ is an integral domain and $Frac(R)$ is its field of fractions, then \reslin{R} is equivalent to \reslin{Frac(R)} and 
\treslin{R} is equivalent to \treslin{Frac(R)}.
\end{proposition}
\begin{proof}
Every proof in \reslin{\R} is also a proof in \reslin{Frac(\R)}. To get the converse, just multiply every line by the least
common multiple (lcm) of all the coefficients in the \reslin{Frac(\R)} proof. If $a_1,\ldots,a_N\in\R$ is the list of denominators of
all the coefficients in a \reslin{Frac(\R)} proof $\pi$, then under a reasonable encoding of \R: 
$|lcm(a_1,\ldots,a_N)|\leq |a_1|+\dots+|a_N|\leq |\pi|$. Therefore the corresponding \reslin{\R} proof is of size at most $O(|\pi|^2)$. 
\end{proof}

\subsection{Basic Counting in \reslin{R} and \reslinsw{R}}

Here we introduce several unsatisfiable sets of linear clauses that  express some counting principles, and serve to exemplify the ability of dag-like \reslin{R}, \treslin{R} and \treslinsw{R} to reason about counting, for a ring \R. We then summarize what we know about refutations of these instance in our different systems, proving along the way some upper bounds and stating some lower bounds proved in the sequel. 

Our unsatisfiable instances are the following: 

\begin{description}
\item[Linear systems:] If $A=(B|b)$ is an $m\times (n+1)$ matrix over $\R$, where the $B$ sub-matrix 

 consists of the first $n$ columns,
such that $B \overline{x}=b$ has no 0-1 solutions, then ($B_i$ is the $i$th row in $B$):
\begin{equation}\label{eq:linSys}
\linSys{A}:=\{B_i\cd \overline{x}=b_i\}_{i\in[m]\,.}
\end{equation}
\item[Subset Sum:] Let  $f$ be a linear form over $\R$ such that $0\notin im_2(f)$. Then,
\begin{equation}\label{eq:subSum}
\subSum{f}:=\{f=0\}\,.
\end{equation}
\item[Image avoidance:] Let  $f$ be a  linear form over $\R$ and recall the notation $\langle f\neq A \rangle$ from Sec.~\ref{sec:notations}. We define
\begin{equation}\label{eq:imAv}
\negIm{f}:=\{\langle f\neq A\rangle: {A\in im_2(f)} \}\,.
\end{equation}
\end{description}

We also consider the following (tautological) generalization of the boolean  axiom $x=0 \vee x=1$.   
\begin{description}
\item[Image axiom:] For $f$ a  linear form, define
\begin{equation}\label{eq:im}
\imAx{f}:=\bigvee\limits_{A\in im_2(f)}f=A\,.
\end{equation}
\end{description}


\medskip 

\subsubsection*{Dag-Like \reslin{R}}

\noindent\uline{Upper bounds.} For any given linear form $f$, $\imAx{f}$ has a \reslin{R}-derivation of polynomial-size (in the size of $\imAx{f}$):
\begin{proposition}\label{imf}
Let $f=\sum\nolimits_{i=1}^na_ix_i+b$ be a linear form over $R$. There exists a \reslin{R} derivation of $\imAx{f}$ of size polynomial in $|\imAx{f}|$ and of principal width 
at most $3$.
\end{proposition}
\begin{proof}

We construct derivations of $\imAx{\sum\nolimits_{i=1}^ka_ix_i+b}$, $0\leq k\leq n$, inductively on $k$.

\Base  $k=0$. In this case $\imAx b$ is just the axiom $b=b$ and thus derived in one step. 

\induction Let $f_k:=\sum\nolimits_{i=1}^ka_ix_i+b$ and assume $\imAx{f_k}$ was already  derived.
Derive $C_0:=\left(\bigvee\nolimits_{A\in im_2(f_k)}f_k+a_{k+1}x_{k+1}=A\right)\vee x_{k+1}=1$ from $\imAx{f_k}$ by
$|im_2(f_k)|$ many resolution applications  with $x_{k+1}=0\vee x_{k+1}=1$. Similarly derive 
$C_1:={\left(\bigvee\nolimits_{A\in im_2(f_k)}f_k+a_{k+1}x_{k+1}=A+a_{k+1}\right)}\vee x_{k+1}=0$ and obtain
$\imAx{f_{k+1}}$ by resolving $C_0$ with $C_1$ on $x_{k+1}$. The size of the derivation is $~n\cd|\imAx{f}|,$ ~and as
there is no clause with more than $3$ equations that  determines  non-parallel
hyperplanes, hence the principal width of the derivation is at most $3$.
\end{proof}

\begin{proposition}\label{ssUB}
For every linear form $f$ such that $0\notin im_2(f)$, the contradiction $\subSum{f}$ admits \reslin{R} refutation of size polynomial in $|\imAx{f}|$.
\end{proposition}
\begin{proof}
First construct the shortest derivation of $\imAx{f}$, and then by a sequence of $|im_2(f)|$ many application of the resolution rule with $f=0$ derive the empty clause.
By Proposition~\ref{imf} the resulting refutation is of polynomial in $|\imAx{f}|$ size. 
\end{proof}

\begin{proposition}\label{eqFromNeq}
Let $f$ be a linear form over $R$, $a\in im_2(f)$ and ${\phi=\{\langle f\neq b\rangle \}_{b\in im_2(f),\,b\neq a}}$. Then there exists
\reslin{R} derivation $\pi$ of $f=a$ from $\phi$, such that $S(\pi)=poly(|\phi|)$ and $\omega_0(\pi)\leq 3$.
\begin{proof}
Let $A_1,\dots,A_N=a$ be an enumeration of all the
elements in $im_2(f)$. By Proposition~\ref{imf} there exists a derivation of 
$\left(\bigvee\nolimits_{i\geq 1}f=A_i\right)$ of principal width at most $3$. For $1<k<N$, we derive ${C:=\left(\bigvee\nolimits_{i\geq k+1}f=A_i\right)}$ from  $\left(\bigvee\nolimits_{i\geq k}f=A_i\right)=(C\vee f=A_k)$
and $\langle f\neq A_k\rangle = (C\vee f=A_1\vee\dots\vee f=A_{k-1})$ in $k-1$ steps as follows: at the $s$th step we get 
$(C\vee f-f=A_s-A_k\vee f=A_{s+1}\vee\dots\vee f=A_{k-1})=(C\vee f=A_{s+1}\vee\dots\vee f=A_{k-1})$ by resolving 
$C\vee f=A_s\vee\dots\vee f=A_{k-1}$ with $C\vee f=A_k$. We thus obtain a derivation of principal width $\omega_0\leq 3$ and of size 
$(1+\dots+(N-2))|f|=\frac{(N-1)(N-2)}{2}|f|$.
\end{proof}
\end{proposition}

\begin{corollary}\label{imAvUB}
For every ring \R\ and every  linear form $f$ the contradiction $\negIm{f}$ admits polynomial-size \reslin{R} refutations.
\end{corollary}
\begin{proof}
Pick some $a\in im_2(f)$. By Proposition~\ref{eqFromNeq} there is a derivation of $f=a$ from $\negIm{f}$ of polynomial size. This derivation
can be extended to a refutation of $\negIm{f}$ by a sequence of resolution rule applications of $f=a$ with $\langle f\neq a\rangle\in \negIm{f}$.
\end{proof}

In Section~\ref{sec:linSysDLUB} we prove an upper bound for $\linSys{A}$ in terms of the size of the image of the affine map,
corresponding to $A$ (Theorem~\ref{thm:linSysImUB}). All other \reslin{\R} upper bounds for $\linSys{A}$  are tree-like. So for more $\linSys{A}$
upper bounds we refer the reader to the \treslin{\R} upper bounds further in this section.
\medskip

\noindent\uline{Lower bounds.} In Sec.~\ref{sec:ssLB} we  prove an exponential lower bound for $\subSum{f}$ in case $f$ is a linear form with
large coefficients (Theorem~\ref{thm:dagLB}).

\subsubsection*{Tree-Like \reslin{R}}

\uline{Upper bounds.} In case $R$ is a finite ring, in Sec.~\ref{sec:NLDT} we prove that the clauses in $\imAx{f}$ admit derivations
of polynomial size (Theorem~\ref{ImUB}). Obviously, in that case ($R$ is finite) any unsatisfiable $R$-linear equation $f=0$ has at most
$|R|$ variables and $\subSum{f}$ are always refutable in constant size. In contrast, in case $R=\Q$ we prove a lower
bound for $\imAx{f}$, $\subSum{f}$ and $\negIm{f}$ for a specific $f$ with small coefficients (see the lower bounds below).

In case a matrix $A=(B|b)$ with entries in a field $\F$ defines a system of equations  $B\overline{x}=b$, that  is unsatisfiable under arbitrary $\F$-valued
assignments (not just under 0-1 assignments), we prove a polynomial upper bound for \treslin{\F} refutations of $\linSys{A}$.
\begin{proposition}\label{linSysUB}
If a $m\times (n+1)$ matrix $A=(B|b)$ with entries in a field $\F$ is such that $B\overline{x}=b$ has no $\F$-valued solutions, then there exists
\treslin{\F} refutation of $\linSys{A}$ of linear size.
\end{proposition}
\begin{proof}
It is a well-known fact from linear algebra that $B\overline{x}=b$ has no $\F$-valued solutions iff there exists $\alpha\in \F^m$ such that 
$\alpha^TB=0$ and $\alpha^T b=1$. Therefore, by $m-1$ resolutions of $B_1\overline{x}-b_1=0,\ldots,B_m\overline{x}-b_m=0$ we can derive 
$-\alpha_1(B_1\overline{x}-b_1)-\ldots-\alpha_m(B_m\overline{x}-b_m)=0$, which is $1=0$. \fedor{mention that $\alpha_i$ are bounded}
\end{proof}

\noindent\uline{Lower bounds.} In Sec.~\ref{sec:ssLB} we prove \treslin{\Q} exponential-size lower bounds
for derivations of $\imAx{f}$ and refutations of $\subSum{f}$ for any $f$ (Corollary~\ref{cor:imTLLB} and Theorem~\ref{thm:ssTLLB}). For $\negIm{f}$
whenever $f$ is of the form $f=\epsilon_1x_1+\ldots+\epsilon_nx_n-A$
for some $\epsilon_i\in\{-1,1\},A\in\F$ the lower bound holds even for the stronger system \treslinsw{\F} (see below).

\subsubsection*{Tree-Like \reslinsw{R}}

\noindent\uline{Upper bounds.} Most of the instances above admit short derivations/refutations in \treslinsw{R}: $\imAx{f}$ is semantic weakening of
$0=0$ and thus derivable in one step; The empty clause is a semantic weakening of $\subSum{f}$ and $\linSys{A}$ and thus can be refuted via deriving 
$\bigvee\nolimits_{i\in[m]}\langle A_i\overline{x}-b_i\neq 0\rangle$
as a semantic weakening of $0=0$ and resolving it with equalities in $\linSys{A}=\{A_i\overline{x}-b_i=0\}_{i\in[m]}$. 

\noindent\uline{Lower bounds.} In case \F\ is a field of characteristic zero,  $\negIm{f}$ are hard even for \treslinsw{R}
whenever $f$ is of the form $f=\epsilon_1x_1+\ldots+\epsilon_nx_n-A$ for some $\epsilon_i\in\{-1,1\},A\in\F$ (Theorem~\ref{thm:imAvLB}).

\subsection{CNF Upper Bounds for \reslin{R} }\label{sec:upper_bounds}
 In this section we outline two basic polynomial upper bounds, which we use to establish  our separations in subsequent sections: short \treslin{R} refutations  for CNF encodings of linear systems 
over a ring $R$, and short \reslin{R} refutations for $\neg\text{PHP}^m_n$. Together with our lower bounds, these imply the separation between \treslin{\F} and \treslin{\F'}, where $\F,\F'$ are fields of positive characteristic such that $char(\F)\neq char(\F')$.  The short refutation of the pigeonhole principle will imply a  separation between dag-like and tree-like \reslin{\F} for fields $\F$ of characteristic 0.

In what follows we consider standard CNF encodings of linear equations $f=0$ where the linear equations are considered as boolean  functions (i.e., functions from 0-1 assignments to $\bits$); we do not use extension variable in these encodings. 

\begin{proposition}\label{cnfLinSysUB}
Let \F\  be a field and $A\overline{x}=b$ be a system of linear equations that has no solution over \F, where $A$ is $k\times n$ matrix with entries in \F, and $A_i$ denotes the $i$th row in $A$. Assume that $\phi_i$ is a CNF encoding of $A_i\cd \overline{x}-b_i=0$, for $i\in[k]$. Then, there exists a \treslin{\F} refutation 
of $\phi=\{\phi_i\}_{i\in[k]}$ of size polynomial in $|\phi|+\sum\nolimits_{i\in[k]}\big|A_i\cd \overline{x}-b_i=0\big|$.
\end{proposition}
\begin{proof}
The idea is to derive the actual linear system of equations from their CNF encoding, and then refute the linear system using a previous upper bound (Proposition~\ref{linSysUB}). 

If $n_i$ is the number of variables in $A_i\cd \overline{x}-b_i=0$, then $|\phi_i|=\Theta(2^{n_i})$. By Proposition~\ref{univUB} proved in the sequel there exists
a \treslin{\F} derivation of $A_i\cd \overline{x}-b_i=0$ from $\phi_i$ of size $O(2^{n_i}|A_i\cd \overline{x}-b_i=0|)=O(|\phi_i|\cd\big|A_i\cd \overline{x}-b_i=0\big|)$.

By Proposition~\ref{linSysUB} there exists a \treslin{\F} refutation of ${\{A_i\cd \overline{x}-b_i=0\}_{i\in[k]}}$ of size $O\left(\sum\nolimits_{i\in[k]}|A_i\cd \overline{x}-b_i=0|\right)$.
The total size of the resulting refutation of $\phi$ is 
$O
\left(
    {\sum\nolimits_{i\in[k]} \big|\phi_i|\cd|A_i\cd \overline{x}-b_i=0\big|}\right)
$ 
and thus is 
$O\left({\left(\sum\nolimits_{i\in[k]}|\phi_i|+\sum\nolimits_{i\in[k]}|A_i\cd \overline{x}-b_i=0|\right)^2}\right)=
O\left({\left(|\phi|+\sum\nolimits_{i\in[k]}|A_i\cd \overline{x}-b_i=0|\right)^2}\right)$.
\end{proof}

As a corollary we get the polynomial upper bound for the Tseitin formulas (see Sec.~\ref{sec:tsDef} for the definition):

\begin{theorem}
Let $G=(V,E)$ be a $d$-regular directed  graph, $p$ a prime number, $\sigma:V\rightarrow \F_p$ such that $\sum\nolimits_{u\in V}\sigma(u)\not\equiv 0 ~({\rm mod}\ p)$,
then $\neg{\rm TS}^{(p)}_{G,\sigma}$ admit \treslin{\text{$\F_p$}} refutations of polynomial size.
\end{theorem}
\begin{proof}
$\neg\text{TS}^{(p)}_{G,\sigma}$ is an unsatisfiable system of linear equations over $\F_p$ (note that no assignment of \F-elements to the variables in $\neg\text{TS}^{(p)}_{G,\sigma}$ is satisfying, and so we do not need to use the (non-linear) boolean  axioms to get the unsatisfiability of the system of equations). Therefore, by Proposition~\ref{cnfLinSysUB} there exists a \treslin{\text{$\F_p$}}
refutation of $\neg\text{TS}^{(p)}_{G,\sigma}$ of polynomial size.
\end{proof}

\begin{theorem}[Raz and Tzameret \cite{RT07}]\label{phpUB}
Let $R$ be a ring such that $char(R)=0$. There exists a \reslin{R} refutation of $\neg PHP^m_n$ of polynomial size.
\end{theorem}
\begin{proof}
This follows from the upper bound of \cite{RT07} for \reslin{\Z} and the fact that any \reslin{\Z} proof can be interpreted as
\reslin{R} if $R$ is of characteristic $0$.
\end{proof}

\section{Dag-Like Lower Bounds}\label{sec:LBDL}\label{sec:ssLB}


In this section we prove an exponential lower bound on the size of dag-like \reslin{\Q} refutations of $\subSum{f}$, where $f=1+x_1+\dots+2^nx_n$.

The lower bound is obtained by defining a mapping, that sends every refutation $\pi$ of $f=0$ to a derivation $\pi'$ from the boolean axioms of some clause $C_{\pi}$, in such a way that $\pi'$ satisfies two properties:

\begin{enumerate}
\item $\pi'$ is at most polynomially larger than $\pi$;
\item $C_{\pi}$ is exponentially large.
\end{enumerate}

We ensure that the second property holds by defining the construction of $\pi'$ in such a way that every disjunct $g=0$ in $C_{\pi}$ has a sufficiently 
small number $Z_g$ of 0-1 solutions, namely $Z_g$ is at most $2^{cn}$, for some constant  $c<1$. This, together with the observation that $C_{\pi}$
must be a boolean  tautology, because it is derivable from the boolean  axioms only, implies that $C_{\pi}$ must be of exponential size (since $C_\pi$ has $2^n$ satisfying assignments and each disjunct contributes at most $2^{cn}$ satisfying disjunctions). Therefore,
by the first property, $\pi$ must be of exponential size.

The fact that $f$ has exponentially large coefficients is essential in our proof that $C_{\pi}$ is of exponential size. All contradictions 
of the form $f=0$, where $f$ has polynomially bounded coefficients, have polynomial dag-like \reslin{\Q} refutations and, thus, there is no
hope to prove strong bounds for dag-like refutations in this case. However, in Sec~\ref{LBTL} we prove that any $f=0$, as long as 
$f$ depends on $n$ variables, must have \treslin{\Q} refutations of size at least $2^{\Omega(\sqrt{n})}$. The argument relies on
a similar transformation from refutations $\pi$ of $f=0$ to derivations of some $C_{\pi}$ and in this way reduces the problem to proving
size lower bounds against \treslin{\Q} derivations of $C_{\pi}$ from the boolean  axioms. 

In order to deal with both tree-like and dag-like lower bounds  we formulate and prove a generalised statement about the translation. For both dag-like and tree-like lower bounds we 
need that for all the disjuncts $g=0$ in $C_{\pi}$ a certain predicate $\mathcal{P}$ holds for $g$. In case of the dag-like bound, 
$\mathcal{P}(g)=1$ iff $g=0$ has at most $2^{cn}$ 0-1 solutions, while in case of the tree-like bound $\mathcal{P}(g)=1$ iff $g$ depends on
at least $\frac{n}{2}$ variables. In Theorem~\ref{thm:refToDer} we prove that the translation can be achieved as long as $\mathcal{P}$
satisfies certain properties (in what follows $\F[x_1,\ldots,x_n]_{\leq 1} $ denotes the linear polynomials in $\F[x_1,\ldots,x_n]$).

\begin{theorem}\label{thm:refToDer}
Let $f$ be a linear polynomial over a field $\F$ with $n$ variables and let $\mathcal{P}:\mathbb{P}(\F[x_1,\ldots,x_n]_{\leq 1})\rightarrow \{0,1\}$
be a predicate on the projective space\footnote{Here, a \emph{projective space} $\mathbb{P}(\F[x_1,\ldots,x_n]_{\leq 1})$ means the set of linear polynomials quotient  by the relation $f\sim \alpha f$ for nonzero scalars $\alpha$.} of linear polynomials over $\F$ satisfying the following properties:
\begin{enumerate}
\item for all linear polynomials $g$ and for all but at most one $a\in\F$: $\mathcal{P}(g+af)=1$;
\item for all $b\in\F$: $\mathcal{P}(b+f)=1$.
\end{enumerate}
If there exists \reslin{\F} (resp.~\treslin{\F}) refutation of $f=0$ of size $S$, then there exists \reslin{\F} (resp.~\treslin{\F})
derivation of size $O(n\cd S^3)$ of a linear clause $\bigvee\nolimits_{j\in[N]}g_j=0$ (for some positive $N$), where $\mathcal{P}(g_j)=1$ for every $j\in[N]$. 
\end{theorem}
\begin{proof}

We now sketch the plan of the proof. Assume that  $\pi$ is a \reslin{\F} refutation   of $f=0$.
By taking out resolutions with $f=0$ we transform $\pi$ into a derivation $\pi'$ of some clause $C$ such that $\mathcal{P}(g)=1$ for every
disjunct $g=0$ in $C$. We do this in such a way that $\pi'$ is not much larger than $\pi$: $|\pi'|=O(n\cd |\pi|^3)$.


Denote $\pi_{\leq k}$ the fragment of $\pi$, consisting of the first $k$ lines of $\pi$. By induction on $k$ we define the sequence $\pi'_{k}$ of derivations
of some clauses $D_k$ from boolean  axioms. The derivations $\pi'_{k}$ are defined together with a surjective function $\tau_k$ from lines of $\pi_{\leq k}$ to
lines of $\pi'_k$ such that if $D=\left(\bigvee\limits_{t\in[m]}g_t=0\right)$ is a line in $\pi_{\leq k}$, then
$$\tau_k(D)=\left(\bigvee\limits_{t\in[m]}g_t+a_tf=0\right)\vee\bigvee\limits_{s\in[m']}h_s=0$$ is a line in $\pi'_k$, where $a_t\in\F$ and
each $h_s$ is a
linear polynomial. Moreover, $\tau_k(D)$ satisfies the following properties:
\begin{enumerate}
\item For each $h_s=0$: $\mathcal{P}(h_s)=1$. 
\item The sets $H_D$ of disjuncts $h_s=0$ in $\tau_k(D)$ are not too large: $\left|\bigcup_{D\in\pi_{\leq k}}H_D\right|\leq 2|\pi_{\leq k}|$. 
\item The numbers $a_t$ and coefficients of $h_s$ are not too large: their bit-size  does not exceed the maximal bit-size of coefficients in $\pi$.
\end{enumerate}

Before we proceed to the inductive definition of $\pi'_k$, we finish the proof assuming that $\pi'_k$ described above exists. 
If $l$ is the length of $\pi$, then $\pi':=\pi'_l$ contains
a derivation of $\tau_l(\emptyset)$, where $\emptyset$ denotes the empty clause. 

We now turn to the inductive definition of $\pi'_k$.
 
\Base Define $\pi'_0$ to be the empty derivation.
\induction Assume $\pi'_k$ and $\tau_k$ satisfy the properties above and $k$ is smaller than the length of $\pi$. If $D$ is the last line of $\pi_{\leq k+1}$,
then $\tau_{k+1}$ extends $\tau_k$ to $D$ and $\pi'_{k+1}$ either extends $\pi'_k$ with $\tau_{k+1}(D)$ or coincides with $\pi'_k$.
Consider the possible cases in which the last
line $D$ of $\pi_{\leq k+1}$ is derived:

\case 1 Boolean  axiom: $D=(x_i=0\vee x_i=1)$. Then $\pi'_{k+1}$ extends $\pi'_k$ with $D$ and $\tau_{k+1}(D)=D$.

\case 2 $D=(f=0)$. Then $\pi'_{k+1}$ extends $\pi'_k$ with the axiom $0=0$ and $\tau_{k+1}(D)=(f-f=0)$.

\case 3 $D$ is derived by resolution: $D=(C_1\vee C_2\vee \alpha G_1+\beta G_2=0)$ for some lines $(C_1\vee G_1=0)$ and $(C_2\vee G_2=0)$ in $\pi_{\leq k}$.

If $C_i=\bigvee\limits_{t\in[m_i]}g^{(i)}_t=0$, by induction hypothesis $\tau_k(C_i\vee G_i=0)$ is of the form ($i=1,2$):
$$\tau_k(C_i\vee G_i=0)=\left(G_i+A_if=0\vee\bigvee\limits_{t\in[m_i]}g^{(i)}_t+a^{(i)}_tf=0\right)\vee\bigvee\limits_{s\in[m'_i]}h^{(i)}_s=0$$

Define $\tau_{k+1}(D)$ to be the following resolution of $\tau_k(C_1\vee G_1=0)\in\pi'_k$ with $\tau_k(C_2\vee {G_2=0})\in\pi'_k$:
\begin{multline*}
\tau_{k+1}(D):=\left(\alpha G_1+\beta G_2+(\alpha A_1+\beta A_2)f=0\vee\bigvee_{i=1,2}\bigvee\limits_{t\in[m_i]}g^{(i)}_t+a^{(i)}_tf=0\right)\vee \\
\vee\bigvee_{i=1,2}\bigvee\limits_{s\in[m'_i]}h^{(i)}_s=0
\end{multline*}
The derivation $\pi'_{k+1}$ extends $\pi'_k$ with $\tau_{k+1}(D)$. It remains to be shown that $\tau_{k+1}(D)$ is of required form and that $\tau_{k+1}$
satisfies the required properties.

If we consider the clause $(\alpha G_1+\beta G_2=0\vee C_1\vee C_2)$ as a \emph{multiset} of disjuncts and $C_1$, $C_2$, as usual, as sets of disjuncts, there can be up to three identical copies of $g=0$ (from $C_1$, from $C_2$ and from $\{\alpha G_1+\beta G_2=0\}$), that are contracted to a single element in the set $D$. In $\tau_{k+1}(D)$ these copies can be different because of different $+af$ terms and, thus, can be non-contractible.

For every disjunct $g=0$ in $D$, denote $\mathcal{F}_g$ the set of disjuncts in $\tau_{k+1}(D)$ that correspond to $g$, namely,
$(g^{(i)}_j+a^{(i)}_jf=0)\in\mathcal{F}_g$ iff $g^{(i)}_j=g$ and $(\alpha G_1+\beta G_2+(\alpha A_1+\beta A_2)f=0)\in\mathcal{F}_g$ iff
$\alpha G_1+\beta G_2=g$. 
For every $g=0\in D$, pick one
element $g+af=0\in\mathcal{F}_g$, which minimises $\mathcal{P}(g+af)$, and denote $X$ the set of these elements. Denote 
$Y:=\left(\bigcup_{g=0\in D}\mathcal{F}_g\right)\backslash X$. Write $\tau_{k+1}(D)$ as follows:
$$\tau_{k+1}(D)=\left(\bigvee\limits_{g+af=0\in X}g+af=0\right)\vee\left(\bigvee_{i=1,2}\bigvee\limits_{s\in[m'_i]}h^{(i)}_s=0\vee
\bigvee\limits_{g+af=0\in Y}g+af=0\right)$$

We now show that $\tau_{k+1}$ satisfies all the desired properties:
\begin{enumerate}
\item For every $h^{(i)}_s=0$, $\mathcal{P}(h^{(i)}_s)=1$ holds by induction hypothesis. 
For every $g+af=0\in Y$, $\mathcal{P}(g+af)=1$ holds by definition of $Y$. 
\item Note that $|H_D\backslash\{h^{(i)}_s=0\}_{i,s}|\leq 2|D|$. By induction hypothesis
 $|\bigcup_{\tilde{D}\in\pi_{\leq k}}H_{\tilde{D}}|\leq 2|\pi_{\leq k}|$.

It follows that 
$|\bigcup_{\tilde{D}\in\pi_{\leq k}}H_{\tilde{D}}\cup H_D|=|\bigcup_{\tilde{D}\in\pi_{\leq k}}H_{\tilde{D}}\cup (H_D\backslash \{h^{(i)}_s=0\}_{i,s})|\leq 
|\bigcup_{\tilde{D}\in\pi_{\leq k}}H_{\tilde{D}}|+|H_D\backslash \{h^{(i)}_s=0\}_{i,s}|\leq 2|\pi_{\leq k}|+2|D|\leq 2|\pi_{\leq k+1}|$.
\item The absolute values of coefficients in $\pi'_{k+1}$ do not exceed the maximal absolute value of coefficients in $\pi$.
\end{enumerate}

\case 4 $D$ is derived by simplification from a line $D\vee b=0$ in $\pi_{\leq k}$. If $D=\left(\bigvee\limits_{t\in[m]}g_t=0\right)$, then $\tau_k(D\vee b=0)$
has the form: $\tau_{k}(D\vee b=0)=\left(\bigvee\limits_{t\in[m]}g_t+a_tf=0\right)\vee b+af=0$.

If $a=0$, we apply simplification to $\tau_{k}(D\vee b=0)$ to derive $\tau_{k+1}(D):=\left(\bigvee\limits_{t\in[m]}g_t+a_tf=0\right)$ and let $\pi'_{k+1}$
extend $\pi'_k$ . 

Otherwise, if $a\neq 0$, we define $\tau_{k+1}(D)$ to be $\tau_{k+1}(D):=\tau_k(D\vee b=0)$ and $\pi'_{k+1}:=\pi'_k$.

\case 5 $D$ is derived by weakening from a line $C$ of $\pi_{\leq k}$: $D=(C\vee g=0)$ for some $g$. Define $\tau_{k+1}(D):=(\tau_k(C)\vee g=0)$ and 
let $\pi'_{k+1}$ extend $\pi'_k$ with $\tau_{k+1}(D)$.
\end{proof}

\begin{lemma}\label{lm:IKtradeoff}
Let $g:\Z^n\rightarrow \Z$ be a linear function. For the sets $I(g):=im_2(g)$ and $K(g):=g^{-1}(0)\cap\{0,1\}^n$ it holds that $|I(g)|\cd |K(g)|\leq 3^n$.
\end{lemma}
\begin{proof}
For every element $a\in I(g)$ choose some $v_a\in\{0,1\}^n$ such that $g(v_a)=a$. Consider the set $X:=\{v_a+u\}_{a\in I(g),u\in K(g)}\subset \{0,1,2\}^n$.

It is easy to see that $|X|=|I(g)|\cd |K(g)|$. Indeed, if $v_a+u=v_{a'}+u'$, then $g(v_a)+g(u)-g(0)=g(v_a+u)=g(v_{a'}+u')=g(v_{a'})+g(u')-g(0)$ and therefore
 $a=a', v_a=v_{a'}, u=u'$.

On the other hand, $|X|\leq 3^n$.
\end{proof}

\begin{lemma}\label{lm:dichotomy}
Let $f=1+2x_1+\dots+2^nx_n$ and $g:\Z^n\rightarrow \Z$ be a linear function. For any 
$a\in\Z\backslash \{0\}$ one of the following holds:
\begin{enumerate}
\item $g=0$ has at most $3^{\frac{n}{2}}$ 0-1 solutions. 
\item $g+af=0$ has at most $3^{\frac{n}{2}}$ 0-1 solutions. 
\end{enumerate}
\end{lemma}
\begin{proof}
For every $b\in\Z$, there exists at most one boolean  assignment that satisfies both $g=b$ and $b+af=0$. Therefore the number of 0-1 solutions of
$g+af=0$ is at most the size of the boolean  image $im_2(g)$ of $g$. By Lemma~\ref{lm:IKtradeoff} either $|im_2(g)|\leq 3^{\frac{n}{2}}$ or
$|g^{-1}(0)\cap\{0,1\}^n|\leq 3^{\frac{n}{2}}$.
\end{proof}

\begin{theorem}\label{thm:dagLB}
Let $f=1+2x_1+\dots+2^nx_n$. Any \reslin{\Q} refutation of $f=0$ is of size $2^{\Omega(n)}$.
\end{theorem}
\begin{proof}
Define the predicate $\mathcal{P}(g)$ on linear polynomials over $\Q$ as follows: $\mathcal{P}(g)=1$ iff $g=0$ has at most 
$2^{(0.5\cd {\log 3})n}$ 0-1 solutions. By Lemma~\ref{lm:dichotomy}, $\mathcal{P}$ satisfies the properties in Theorem~\ref{thm:refToDer}.
Therefore, by Theorem~\ref{thm:refToDer}, if $\pi$ is a refutation of $f=0$, then there exists a derivation $\pi'$ of some clause $C=\bigvee\limits_{j\in[N]}g_j=0$ from the boolean  axioms, where each $g_j=0$ has at most $2^{(0.5\cd {\log 3})n}$ 0-1 solutions. Moreover $|\pi'|=O(n\cd |\pi|^3)$. As $C$ must be a boolean  tautology, that satisfied by $2^n$ assignments, it must contain at least $2^{(1-0.5\cd {\log 3})n}$ disjuncts (because every disjunct contributes at most $2^{(0.5\cd {\log 3})n}$ satisfying assignments). Therefore $|\pi|=2^{\Omega(n)}$.
\end{proof}

\section{Linear Systems with Small Coefficients}
\label{sec:Linear-Systems-with-Small-Coefficients}

In this section we study 0-1 unsatisfiable linear systems over \emph{finite fields}.

Firstly, we prove an upper bound, which is polynomial in 
$|im_2(A\, \overline{x})|$, where $A=A_{f_1,\dots,f_m}:\F^n\rightarrow \F^m$ is an affine map $\overline{x}\mapsto (f_1(\overline{x}),\ldots,f_m(\overline{x}))$.
In contrast to the case of a single equation $f=0$, the size of the image $|im_2(A\, \overline{x})|$ does not fully characterise the size of the shortest
\reslin{\F} refutation of $f_1=0,\ldots,f_m=0$: there is an example, where $|im_2(A\, \overline{x})|$ is large, but the size for refuting $f_1=0,\ldots,f_m=0$ 
is small.

Secondly, we prove a superpolynomial lower bound on a linear system for a \emph{restricted} tree-like \reslin{\F}.

\subsection{An Upper Bound}\label{sec:linSysDLUB}

Denote $\langle A_{f_1,\ldots,f_m}\, \overline{x}\neq 0\rangle$ the linear clause 
$\left(\langle f_1\neq 0\rangle\vee\dots\vee \langle f_m\neq 0 \rangle\right)$.
The clause $\langle A_{f_1,\ldots,f_m}\, \overline{x}\neq 0\rangle$ is a tautology iff the system $f_1=0,\ldots,f_m=0$ is 0-1 unsatisfiable. Therefore, any 0-1
unsatisfiable system $f_1=0,\ldots,f_m=0$ can be refuted by first deriving $\langle A_{f_1,\ldots,f_m}\, \overline{x}\neq 0\rangle$ from boolean  axioms
and then resolving it with $f_1=0,\ldots,f_m=0$. We now prove an upper bound for derivations of $\langle A\, \overline{x}\neq 0 \rangle$ in terms of
$|im_2(A\, \overline{x})|$.

\begin{theorem}\label{thm:linSysImUB}
Let $f_1=0,\ldots,f_m=0$ be a 0-1 unsatisfiable system with $n$ variables. There exists a derivation of 
$\langle A_{f_1,\ldots,f_m}\, \overline{x}\neq 0\rangle$ of size $poly(n+|im_2(A_{f_1,\ldots,f_m}\, \overline{x})|)$. 
\end{theorem}
\begin{proof}
We arrange the derivation in $n$ layers $L_0,\ldots,L_n$ in such a way that $L_0:=\{\langle A_{f_1,\ldots,f_m}\, \overline{x}\neq 0\rangle\}$ and 
$$L_k:=\{\left(\langle f_1\rst_{x_1\leftarrow \epsilon_1,\ldots, x_k\leftarrow \epsilon_k}\neq 0\rangle\vee\ldots\vee
 \langle f_m\rst_{x_1\leftarrow \epsilon_1,\ldots, x_k\leftarrow \epsilon_k}\neq 0\rangle\right)\}_{\overline{\epsilon}\in\{0,1\}^k}$$
It is easy to see, that the following map is an embedding $L_k\hookrightarrow im_2(A_{f_1,\ldots,f_m}\, \overline{x})$: 
\begin{multline*}
\left(\langle f_1\rst_{x_1\leftarrow \epsilon_1,\ldots, x_k\leftarrow \epsilon_k}\neq 0\rangle\vee \ldots\vee
 \langle f_m\rst_{x_1\leftarrow \epsilon_1,\ldots, x_k\leftarrow \epsilon_k}\neq 0\rangle\right)\mapsto \\
\left(f_1(\epsilon_1,\ldots, \epsilon_k, 0,\ldots, 0), \ldots,
 f_m(\epsilon_1,\ldots, \epsilon_k, 0,\ldots, 0)\right)
\end{multline*}
Therefore $|L_k|\leq |im_2(A_{f_1,\ldots,f_m}\, \overline{x})|$.

It remains to note that every clause in $L_k$ can be derived from clauses in $L_{k+1}$ in $O(|im_2(A_{f_1,\ldots,f_m}\, \overline{x})|)$ steps.
Indeed, if $C\in L_k$, then $C\rst_{x_{k+1}\leftarrow 0}\in L_{k+1}$ and $C\rst_{x_{k+1}\leftarrow 1}\in L_{k+1}$, and $C$ can be derived from
$C\rst_{x_{k+1}\leftarrow 0}$ and $C\rst_{x_{k+1}\leftarrow 1}$ and the axiom $(x_{k+1}=0\vee x_{k+1}=1)$ in a standard way.
\end{proof}

\begin{remark}
In contrast to the case of a single equation, dag-like \reslin{\F} refutations of $f_1=0,\ldots,f_m=0$ for $m\geq 2$ are not lower-bounded by
$|im_2(A_{f_1,\ldots,f_m}\, \overline{x})|$ in general. For example, the system $x_1-2x_{n+1}=0,x_n-2x_{2n}=0,x_{2n+1}+x_{n+1}+\ldots+x_{2n}-2=0$
has refutation of size $O(n)$, but $|im_2(A_{f_1,\ldots,f_m}\, \overline{x})|=2^{\Omega(n)}$.
\end{remark}

\subsection{Lower Bound for Restricted Tree-Like \reslin{\F}}

We define the following natural model of decision trees, certifying 0-1 unsatisfiability of linear systems over \F:

\begin{definition}\label{def:simpleDT}
Let $A\, \overline{x}=\overline{b}$ be a 0-1 unsatisfiable linear system over $\F$. A decision tree $T$ for $A\, \overline{x}=\overline{b}$ is
a binary tree, such that:
\begin{itemize}
\item Every internal node is labelled with a variable $x_i$ and two branches correspond to assignments $x_i\leftarrow 0$ and $x_i\leftarrow 1$.
\item If $\rho_v$ is the variable assignment made along the path from the root to a leaf $v$, the system $(A\, \overline{x}=\overline{b})\rst_{\rho_v}$
is unsatisfiable over the whole field \F\ (not just over 0-1).
\end{itemize}
\end{definition}

It is easy to see that this model of decision trees can be simulated by tree-like \reslin{\F}. We argue that this model captures the strength of a natural
fragment of tree-like \reslin{\F}. If $T$ is a decision tree for the system $f_1=0,\ldots,f_m=0$ then a corresponding tree-like proof $\pi$
for every leaf $v$ in $T$ derives the set of clauses
$$\left\{\left(f_k\rst_{\rho_v}=0\vee\bigvee\limits_{i\in[n]|\rho_v(i)\neq *}x_i=1-\rho_v(i)\right)\right\}_{k\in[m]}$$
where $\rho_v:[n]\mapsto \{0,1,*\}$ ($\rho_v(i)=*$ iff $x_i$ is unassigned) is the assignment at $v$. By the leaf condition in 
Definition~\ref{def:simpleDT} the system $f_1\rst_{\rho_v}=0,\ldots,f_m\rst_{\rho_v}=0$ is unsatisfiable over \F, therefore there exist
$a_1,\ldots,a_m\in\F$ such that $a_1f_1\rst_{\rho_v}+\cdots+a_mf_m\rst_{\rho_v}=1$ and the proof $\pi$ uses this to derive further the
clause $\bigvee\limits_{i\in[n]|\rho_v(i)\neq *}x_i=1-\rho_v(i)$ from the clauses above for every leaf $v$. This is the \emph{only place}, where
counting is essentially used in $\pi$, the rest of the proof is just a standard resolution refutation obtained from $T$ by the well-known correspondence
between decision trees and tree-like resolution refutations. It is an interesting question whether this fragment is strictly weaker than  full
tree-like \reslin{\F}.    

We now prove a sub-exponential lower bound for this model and, consequently, for the
corresponding fragment of tree-like \reslin{\F}.

\begin{theorem}\label{thm:linSysDTLB}
For every $n\in\N$ there exists a 0-1 unsatisfiable linear system $A\, \overline{x}=\overline{b}$ over a finite field $\F_q,q > 2$ with $n$ variables
such that any decision tree for this system is of size $2^{\Omega\left(\frac{n}{\log n}\right)}$.
\end{theorem}
\begin{proof}
We construct the matrix $A$ as a generator matrix of a linear $(n,k,d)_q=(n,\frac{n}{\log q}+1,\Omega(\frac{n}{\log n}))_q$ error-correcting code 
(Definition~\ref{def:ECC}). 

The condition $k>\frac{n}{\log q}$, which this code satisfies, assures that $q^k>2^n$ and therefore there exists $\overline{b}\in \F_q^k$ such that
$A\, \overline{x}=\overline{b}$ is 0-1 unsatisfiable. 

Note that depths of all leaves in any decision tree for $A\, \overline{x}=\overline{b}$ are at least $d$. Indeed, if $k<d$ variables are substituted
at $v$ by $\rho_v$, then the minimal distance of the code, generated by $A\rst_{\rho_v}$, is at least $d-k$ and, in particular, $A\rst_{\rho_v}$ has
full rank, therefore $v$ is not a leaf. Thus any decision tree for $A\, \overline{x}=\overline{b}$ has size at least $2^d=2^{\Omega(\frac{n}{\log n})}$.

The existence of such a code is guaranteed by the Gilbert bound (Theorem~\ref{thm:Gilbert}). Recall that the Gilbert bound claims the existence of a linear $(n,k,d)_q$ code whenever
$$\sum\limits_{i=1}^d\left(\begin{array}{c}n \\ i\end{array}\right)\cd (q-1)^i<q^{n-k+1}$$
holds. In our case, if we assign $d=\frac{n}{10\log n}$:
\begin{multline*} 
\sum\limits_{i=1}^d\left(\begin{array}{c}n \\ i\end{array}\right)\cd (q-1)^i< d\cd q^{\frac{d\log n}{\log q}}\cd q^d
\leq \frac{n}{10\log n}\cd q^{n(\frac{1}{10\log q}+\frac{1}{\log n})} < q^{n(1-\frac{1}{\log q})+1}\,.
\end{multline*}
\end{proof}

\section{Tree-Like Lower Bounds}\label{LBTL}

\subsection{Nondeterministic Linear Decision Trees}\label{sec:NLDT}
In this section we extend the classical correspondence between tree-like resolution refutations and decision trees (cf.~\cite{BKS04}) to tree-like \reslin{R} and tree-like \reslinsw{R}.
We define \emph{nondeterministic linear decision trees} (NLDT), which generalize parity decision trees, proposed in \cite{IS14} for $R=\F_2$,
to arbitrary rings.
We shall use these trees in the sequel to establish some of our upper and lower bounds (though not for our dag-like lower bounds). 

Let $\phi$ be a set of linear clauses (that we wish to refute) and $\Phi$ a set of linear non-equalities over $R$ (that we take as assumptions).
Consider the following two decision problems: 

\begin{itemize}
  \item[DP1] Assume $\Phi\models\neg\phi$. Given a satisfying boolean  assignment $\rho$ to $\Phi$, determine which clause $C\in\phi$ is violated by $\rho$
             by making queries of the form: which of $f|_{\rho}\neq 0$ or $g|_{\rho}\neq 0$ hold for linear forms $f,g$ in case $f|_{\rho}+g|_{\rho}\neq 0$.
  \item[DP2] Similar to DP1, only that we  assume $\Phi\models_R\neg\phi$, and given $R$-valued assignment $\rho$, satisfying $\Phi$, we ask to find a clause $C\in\phi$
             falsified by $\rho$.
\end{itemize}

Below we define NLDTs of types $\DTsw{R}$ and $\DT{R}$, which provide solutions to DP1 and DP2, respectively. The root of a tree is labeled with a system $\Phi$, the edges in a tree are labeled with linear
non-equalities of the form $f\neq 0$ and the leaves are labeled with clauses $C\in\phi$. Informally, at every node $v$ there is a set $\Phi_v$
of all \emph{learned} 
non-equalities, which is
the union of $\Phi$ and
 the set of non-equalities along the path from the root to the node.
 If $v$ is
an internal node, two outgoing edges $f\neq 0$ and $g\neq 0$ define a query to be made at $v$, where $f+g\neq 0$ is a consequence of
$\Phi_v$. If $v$ is a leaf, then $\Phi_v\cup\Phi$ contradicts a clause $C\in \phi$.

Starting from the root, based on the assignment $\rho$, we go along a path, from the root to a leaf, by choosing in each node to go along the left edge $f\neq 0$ or the right edge $g\neq 0$, depending on whether $f|_{\rho}\neq 0$ or $g|_{\rho}\neq 0$. Note that  $f|_{\rho}\neq 0$ and $g|_{\rho}\neq 0$ may not be mutually exclusive, and this is why the decision made in each node may be \emph{nondeterministic}.    


\begin{definition}[Nondeterministic linear decision tree NLDT; $\DT{R}\label{DTdef}$, $\DTsw{R}$] 
Let $\phi$ be a set of linear clauses and $\Phi$ be a set of linear non-equalities over a ring $R$.
A nondeterministic linear decision tree $T$ of
type $\DT{R}$ and of type $\DTsw{R}$ for $(\phi,\Phi)$ is a binary rooted tree, where every edge is labeled with some linear non-equality $f\neq0$,
in such a way that the conditions below hold.
In what follows, for a node $v$, we denote by $\Phi_{r\leadsto v}$ the set of non-equalities along the path from the root $r$ to $v$ and by 
$\Phi_v$ the set $\Phi_{r\leadsto v}\cup\Phi$. We say that $\Phi_v$ is
the set of \emph{learned} non-equalities at $v$.  
\begin{enumerate}
\item Let $v$ be an internal node. Then $v$ has two outgoing edges labeled by linear non-equalities $f_v\neq 0$ and $g_v\neq 0$, such that:
\begin{itemize}
    \item If $T\in \DT{R}$, then $\alpha f_v+\beta g_v\neq 0\in \Phi_v\cup {\{a\neq 0\;|\; a\in R\setminus 0\}}$ for some $\alpha,\beta\in R$.

    \item If $T\in \DTsw{R}$, then $\Phi_v\models \alpha f_v+\beta g_v\neq 0$ for some $\alpha,\beta\in R$.
\end{itemize}
 
\item\label{leafCond} A node $v$ is a leaf if
there is a linear clause $C\in\phi\cup\{0=0\}$ which is violated by $\Phi_v$ in the following sense:
        \begin{itemize}
         \item  If $T\in \DT{R}$, then $\neg C\subseteq \Phi_v\cup\{a\neq 0\; |\;a\in R\setminus 0\}$.
         \item If  $T\in\DTsw{R}$, then $\Phi_v\models \neg C$.       
        \end{itemize}
%
\end{enumerate}
\end{definition}

In case $\Phi$ is empty, we sometimes simply write that the NLDT is for $\phi$ instead of $(\phi,\emptyset)$.

Assume $\Phi\models\neg\phi$. Then an NLDT for $(\phi\cup\{x=0\vee x=1\,|\, x\in vars(\phi)\},\Phi)$ of type $\DT{R}$ can be converted into an NLDT of type $\DTsw{R}$
for $(\phi,\Phi)$ by truncating all maximal subtrees with all leaves from $\{x=0\vee x=1\,|\, x\in vars(\phi)\}$ and marking their roots with arbitrary clauses from $\phi$.

Below we give several examples (and basic properties) of NLDTs.

\para{Example 1}
Let $\phi$ be a set of clauses, representing unsatisfiable CNF. Then any standard decision tree on boolean  variables is an NLDT for 
$\phi\cup\{x=0\vee x=1\,|\,x\in vars(\phi)\}$ of type $\DT{R}$, where
a branching on the value of a variable $x$ is realized by branching on $(1-x)+x\neq 0$ to either  $1-x\neq 0$ or $x\neq 0$.
This is illustrated by (the proof of) the following proposition:

\begin{proposition}\label{univDT}
If $\Phi$ is a set of linear non-equalities and $\phi$ is a set of linear clauses over $\R$ such that $\Phi\models \neg \phi$, then there exists a $\DT{R}$
tree for $(\phi\cup\{x=0\vee x=1\,|\,x\in vars(\phi\cup\{\neg\Phi\})\},\Phi)$ of size $O(2^n|\Phi|)$, where $n=|vars(\phi\cup\{\neg\Phi\})|$.
\end{proposition}
\begin{proof}
Let $vars(\phi\cup\{\neg\Phi\})=\{x_1,\dots,x_n\}$ and fix an ordering on these variables. Construct a tree $T_0$ with $2^n$ nodes, that  branches on
$x_1,\dots,x_n$, in this order. Thus, in  every leaf $v$ of $T_0$ a total assignment to the variables is determined (i.e., $\Phi_v=\{x_i\neq \nu_i\}_{i\in [n]}\cup\Phi$ for some $\nu_i\in\{0,1\}$). 
Since $\Phi\models \neg\phi$, this assignment violates either some clause $C=(f_1=0\vee\dots\vee f_m=0)$ in $\phi$ or some non-equality $g\neq 0$ in $\Phi$. 
We  augment $T_0$ to $T$ by attaching a subtree to every leaf $v$ of $T_0$ depending on whether the former or latter condition holds for $v$, as follows:

\case 1 $\{x_i\neq \nu_i\}_{i\in [n]}\models \neg C$. We attach a subtree to $v$ that makes $m$  sequences of branches as follows. If $f_i=a_1x_1+\ldots+a_nx_n+b$ then 
$a_1(1-\nu_1)+\ldots+a_n(1-\nu_n)+b\neq 0$ holds and the $i$th sequence is the following sequence of ``substitutions'':
$(a_1x_1+a_2(1-\nu_2)+\ldots+a_n(1-\nu_n)+b)+(a_1(1-\nu_1)-a_1x_1)\neq 0$ to $a_1x_1+a_2(1-\nu_2)+\ldots+a_n(1-\nu_n)+b\neq 0$ and $a_1(1-\nu_1)-a_1x_1\neq 0$,
\ldots, $(a_1x_1+\ldots+a_{n-1}x_{n-1}+a_n(1-\nu_n)+b)+(a_n(1-\nu_n)-a_nx_n)\neq 0$ to $f_i\neq 0$ and $a_n(1-\nu_n)-a_nx_n\neq 0$. All the right branches
lead to nodes $u$ such that $\{x_i\neq 0, x_i\neq 1\}\subseteq \Phi_u$ for some $i\in[n]$ and thus they satisfy the $\DT{R}$ leaf condition in Definition \ref{DTdef}. Such a sequence indeed performs substitutions: the edge to the leftmost node is $f_i\neq 0$ and as we go upwards, we 
apply the substitutions $x_n\leftarrow 1-\nu_n$, \ldots, $x_1\leftarrow 1-\nu_1$ to this non-equality.

In the leftmost node $w$ in the end of the $m$th sequence, $\{f_1\neq 0,\dots,f_m\neq 0\}\subseteq \Phi_w$ holds and thus again $C$ is violated at $w$
in the sense of  Definition~\ref{DTdef} and therefore $w$ is a legal $\DT{R}$-leaf.

\case 2 $\{x_i\neq \nu_i\}_{i\in [n]}\models g=0$, where $g\neq 0\in \Phi_v$. Let $g=a_1x_1+\ldots+a_nx_n+b$. Attach to $v$ a subtree that makes the following
branches: $(a_1(1-\nu_1)+a_2x_2+\ldots+a_nx_n+b)-(a_1(1-\nu_1)-a_1x_1)\neq 0$ to $(a_1(1-\nu_1)+a_2x_2+\ldots+a_nx_n+b)\neq 0$ and
$a_1(1-\nu_1)-a_1x_1\neq 0$,\ldots, $(a_1(1-\nu_1)+\ldots+a_{n-1}(1-\nu_{n-1})+a_n(1-\nu_n)+b)-(a_n(1-\nu_n)-a_nx_n)\neq 0$ to $1\neq 0$ and $a_1(1-\nu_1)-a_1x_1\neq 0$. All leaves of the subtree
satisfy the condition for $\DT{R}$ leaves in Definition~\ref{DTdef}.

The tree $T$ is a $\DT{R}$ tree for $(\phi,\Phi)$.
\end{proof}

\para{Example 2}
Let $\phi$ be as in Example 1. \textit{Parity decision trees}, as defined in \cite{IS14}, are NLDTs for $\phi$ of type $\DTsw{\F_2}$:
branching on the value of an $\F_2$-linear form $f$ is realized by branching from  $(1-f)+f\neq 0$ to $1-f\neq 0$ and $f\neq 0$. And the converse also holds: a branching
of $f+g\neq 0$ to $f\neq 0$ and $g\neq 0$, where, say, $f$ is a non-constant $\F_2$-linear form, is equivalent to branching on the value of $f$.

\para{Example 3}
Let $\phi=\{f_1=0,\dots,f_m=0\}$, where $f_1,\dots,f_m$ are $R$-linear forms such that $f_1+\ldots+f_m=1$. Then a polynomial-size NLDT of type $\DT{R}$ for $\phi$
makes the following branchings, where all right edges lead to a leaf: $(f_1+\ldots+f_{m-1})+f_m\neq 0$ (this is just $1\neq 0$) to $f_1+\ldots+f_{m-1}\neq 0$
and $f_m\neq 0$, \ldots, $f_1+f_2\neq 0$ to $f_1\neq 0$ and $f_2\neq 0$.

\bigskip

We now show the equivalence between NLDTs and tree-like \reslin{R} proofs.

\begin{theorem}\label{treeDTequiv}
Let $\phi$ be a set of linear clauses over a ring \R\ and $\Phi$ be a set of linear non-equalities over \R. Then, there exist decision trees
$\DT{R}$ (resp.~$\DTsw{R}$)  for $(\phi\cup\{x=0\vee x=1\,|\,x\in vars(\phi)\},\Phi)$ (resp. $(\phi,\Phi)$) of size $s$ iff there exist  \treslin{R}
(resp.~\treslinsw{R}) derivations of the clause $\neg\Phi=\bigvee\nolimits_{f\neq 0\in \Phi}f=0$ from $\phi$ of size $O(s)$.
\end{theorem}
\begin{proof}
$(\Rightarrow)$ 
Let $T_{\phi}$ be an NLDT of type $\DT{R}$ or $\DTsw{R}$ for $\phi$. We construct a \treslin{R} or
\treslinsw{R} derivation from $T_{\phi}$, respectively, as follows. Consider the tree of clauses $\pi_0$, obtained from $T_{\phi}$ by replacing every vertex $u$
 with the clause
$\neg\Phi_u$. This tree is not a valid tree-like derivation yet. We augment it to a valid derivation $\pi$ by appropriate insertions of 
applications of weakening and simplification rules. 

\medskip

\case 1 If $\neg \Phi_u\in\pi_0$ is a leaf, then $\Phi_u$ violates a clause $D\in\phi\cup\{0=0\}$.
By condition~\ref{leafCond} in Definition~\ref{DTdef},  $\neg\Phi_u$ must be a weakening of $D$  
(syntactic for $T_{\phi}\in\DT{R}$ and semantic for $T_{\phi}\in\DTsw{R}$) and we add $D$ as the only child of this node.

\medskip

\case 2 
Let $\neg \Phi_u\in\pi_0$ be an internal node with two outgoing edges labeled with $f_u\neq 0$ and $g_u\neq 0$. 

If $T_{\phi}\in\DT{R}$, then $\alpha f_u+\beta g_u\neq 0\in \Phi_u\cup {\{a\neq 0\, |\, a\in R\setminus 0\}}$.
Apply resolution to $\neg\Phi_{l(u)}=(\neg\Phi_u\vee f_u=0)$ and $\neg\Phi_{r(u)}=(\neg\Phi_u\vee g_u=0)$ to derive 
$\neg\Phi_u\vee \alpha f_u+\beta g_u=0$. In case $\alpha f_u+\beta g_u\neq 0\in\Phi_u$ this clause coincides with $\neg\Phi_u$
and no additional steps are required. In case $\alpha f_u+\beta g_u\neq 0\in {\{a\neq 0\, |\, a\in R\setminus 0\}}$ insert an application
of the simplification rule to get a derivation of $\neg\Phi_u$.

If $T_{\phi}\in\DTsw{R}$, $\Phi_u\models \alpha f_u+\beta g_u\neq 0$, we derive $\neg\Phi_u\vee\alpha f_u+\beta g_u=0$
from $\neg\Phi_{l(u)}=(\neg\Phi_u\vee f_u=0)$ and $\neg\Phi_{r(u)}=(\neg\Phi_u\vee g_u=0)$ by an application of the resolution rule and
then deriving $\neg\Phi_u$ by an application of the semantic weakening rule.
\bigskip

\noindent$(\Leftarrow)$
Conversely, assume $\pi$ is a \treslin{R} or a \treslinsw{R} derivation of a (possibly empty) clause $\mathcal{C}$ from $\phi$.
In what follows, when we say weakening we mean syntactic or semantic weakening depending on $\pi$ being  a \treslin{R} or a \treslinsw{R} derivation,
respectively.

Let the edges in the proof-tree of $\pi$ be directed from conclusion to premises.
We turn this proof-tree into a decision tree $T_{\pi}$ for $(\phi,\neg \mathcal{C})$ as follows. Every node of outgoing degree $2$ in the proof-tree
$\pi$ is a clause
obtained from its children by a resolution rule. For each  such node $C\vee D\vee (\alpha f+\beta g=0)$ we label its
outgoing edges to $C\vee f=0$ and $D\vee g=0$ with $f\neq 0$ and $g\neq 0$, respectively. 
We contract all unlabeled edges, which are precisely those corresponding to applications of weakening and simplification rules.
If $C_1,\dots,C_k$ is a  maximal (with
respect to inclusion) sequence of weakening and simplification rule applications (the latter occur only in \reslin{R} derivations), then we contract it to $C_k$.
 In this way we obtain the tree $T_{\pi}$, where every edge
is labeled with linear non-equality and every node $u$ is labeled with a clause $C_u$ such that if $f\neq 0$ and $g\neq 0$ are
labels of edges to the  left $l(u)$ and to the right $r(u)$ children respectively, then $C_u$ is a weakening and a simplification (the latter 
again in case of \reslin{R})
of the clause
$C\vee D\vee\alpha f+\beta g=0$ for some
$\alpha, \beta\in R$, such that $C_{l(u)}=(C\vee f=0)$, $C_{r(u)}=(D\vee g=0)$.

We now prove that $T_{\pi}$ is a valid decision tree of type $\DT{R}$ (respectively, $\DTsw{R}$) if $\pi$ is a \treslin{R}
derivation (respectively, \treslinsw{R} derivation). 
\medskip 

\case 1
Assume $\pi$ is \treslin{R} derivation. We prove inductively that for every node $u$ in $T_{\pi}$ we have $\neg C_u\subseteq \Phi_u$. 
\Base $u$ is the root $r$. We have $\Phi_r=\neg\mathcal{C}=\neg C_r$.
\induction For any other node $u$ assume $\neg C_p\subseteq \Phi_p\cup\{a\neq 0\,|\,a\in R\setminus 0\}$ holds for its parent node $p$. Let $f\neq 0$ be the 
label on the edge from $p$ to $u$. Then $C_u=(C\vee f=0)$ for some clause $C$ and $C_p$ must be of the form $(C\vee D)$
for some clause $D,$ and hence $\neg C_u\subseteq \neg C\cup\{f\neq 0\}\subseteq \neg C_p\cup\{f\neq 0\}\subseteq \Phi_p
\cup \{f\neq 0\} =\Phi_u $. 

Now we show that $T_{\pi}$ satisfies the conditions of Definition \ref{DTdef} for $\DT{R}$ trees.
\begin{itemize}
\item(Internal nodes) Let $u$ be an internal node of $T_{\pi}$ with outgoing edges labeled with $f\neq 0$ and $g\neq 0$.
$C_u$ must be both a weakening and a simplification of $(C\vee\alpha f+\beta g=0)$ for some $\alpha,\beta\in R$ and a linear clause $C$.
If $\alpha f+\beta g\neq 0\in\{a\neq 0\,|\,a\in R\setminus 0\},$ then the condition trivially holds, otherwise $\alpha f+\beta g=0$
cannot be eliminated via simplification and thus $\alpha f+\beta g\neq 0\in\neg C_u$ and $\neg C_u\subseteq \Phi_u$ imply
$\alpha f+\beta g\neq 0\in \Phi_u$ and the condition for internal nodes in Definition \ref{DTdef} is satisfied.
\item(Leaves) Let $u$ be a leaf of $T_{\pi}$. Then $C_u$ must be both a weakening and a simplification of some clause $C$ in 
$\phi\cup\{x=0\vee x=1\,|\,x\in vars(\phi)\}\cup\{0=0\}$,
that is $C_u=(C\vee D)$ for some
clause $D$. Therefore $\neg C_u\subseteq \Phi_u$ implies that $C$ is falsified by $\Phi_u$.
\end{itemize}

\case 2 
Assume $\pi$ is a \treslinsw{R} derivation. We prove inductively that for every node $u$ in $T_{\pi}$,  
$C_u\models \neg\Phi_u$ holds. 
\Base $u$ is the root $r$ and we have $\neg\Phi_r=\mathcal{C}= C_r$.
\induction  $u$ is a node which is not the root. If $C_p\models \neg\Phi_p$ holds for its parent $p$ and $f\neq 0$ is 
the label on the edge from $p$ to $u$, then $(C\vee D\vee \alpha f+\beta g=0)\models C_p$, $C_u=(C\vee f=0)$ for some $\alpha,\beta\in R$ a
linear form $g$ and some linear clauses $C,D$. Therefore, $C_u=(C\vee f=0)\models (C_p\vee f=0)\models (\neg\Phi_p
\vee f=0)=\neg\Phi_u$.

We now show that $T_{\pi}$ satisfies the conditions of Definition~\ref{DTdef} for $\DTsw{R}$ trees.
\begin{itemize}
\item(Internal nodes) Let $u$ be an internal node of $T_{\pi}$ with outgoing edges labeled with $f\neq 0$ and $g\neq 0$. Then 
$(C\vee\alpha f+\beta g=0)\models C_u$ for some $\alpha,\beta\in R$ and a linear clause $C$. Therefore $C_u\models \neg\Phi_u$ implies
$\Phi_u\models\alpha f+\beta g\neq 0$.
\item(Leaves) Let $u$ be a leaf of $T_{\pi}$. Then $C_u$ must be a weakening of some clause $C$ in $\phi\cup\{0=0\}$, that is, $C_u=(C\vee D)$ for some
clause $D$. Therefore $C_u\models \neg\Phi_u$ implies that $C$ is falsified by $\Phi_u$.
\end{itemize}
\end{proof}

An immediate corollary is the following:

\begin{proposition}\label{univUB}
If $\phi\cup\{C\}$ is a set of linear clauses over a ring $R$ such that $\phi\models C$, then there exists a \treslin{R} derivation of $C$ from $\phi$ of size $O(2^n|C|)$,
where $n=\big|vars(\phi\cup\{C\})\big|$.
\end{proposition}
\begin{proof}
By Proposition~\ref{univDT} there exists a $\DT{R}$ tree for $(\phi\cup\{x=0\vee x=1\,|\,x\in vars(\phi\cup\{C\})\},\neg C)$ of size $O(2^n|C|)$ and, thus, by Theorem~\ref{treeDTequiv} there exists
a \treslin{R} derivation of $C$ from $\phi$ of size $O(2^n|C|)$.
\end{proof}

We construct an NLDT to prove the following upper bound:

\begin{proposition}\label{ImUB}
Let $R$ be a finite ring, $f=a_1x_1+\dots+a_nx_n$ a linear form over $R$, $s_f$ the size of $\imAx{f}$ (i.e., the size of its encoding) and $d_f=|im_2(f)|$. Then, there exists a \treslin{R}
derivation of $\imAx{f}$ of size $O(s_fn^{2d_f})$.
\end{proposition}
\begin{proof}
We construct a decision tree of type $\DT{R}$ of size $O(s_fn^{2d_f})$ with the system $\Phi_r=\{f\neq A\}_{A\in im_2(f)}$ at its root $r$. By 
Theorem~\ref{treeDTequiv} this implies the existence of  a \treslin{R} proof of $\imAx{f}$ of the same size.

Let $f^{(1)}:=a_1x_1+\dots+a_{\lfloor\frac{n}{2}\rfloor}x_{\lfloor\frac{n}{2}\rfloor}$ and 
$f^{(2)}:=a_{\lfloor\frac{n}{2}\rfloor+1}x_{\lfloor\frac{n}{2}\rfloor+1}+\dots+a_nx_n$.
The decision tree for $\imAx{f}$ is constructed recursively as a tree of height $2d_f$, where a subtree  for $\imAx{f^{(1)}}$ or for 
$\imAx{f^{(2)}}$ is hanged from each leaf.
At every node $u$ of depth $d$ the system of non-equalities is of the form: $\Phi_u=\Phi_r\cup\Phi_u^{(1)}\cup\Phi_u^{(2)}$, where 
$\Phi_u^{(i)}\subseteq \{f^{(i)}\neq A\}_{A\in im_2(f^{(i)})},~i\in\{1,2\}$ and $|\Phi_u^{(1)}|+|\Phi_u^{(2)}|=d$. A node $u$ is a leaf if and
only if $\Phi_u^{(i)}=\{f^{(i)}\neq A\}_{A\in im_2(f^{(i)})}$ for some $i\in\{1,2\}$. The branching at an internal node $u$ is made by the non-equality
$f^{(1)}-A_1+f^{(2)}-A_2\neq 0$, for some  $A_i\in im_2(f^{(i)})$ where $f^{(i)}-A_i\notin \Phi_u^{(i)},i \in \{1,2\}$. The size $s_n$ of this
tree can be upper bounded as follows: $s_n\leq 2^{2d_f}s_{\lfloor\frac{n}{2}\rfloor+1}+s_f2^{2d_f}=O(s_fn^{2d_f})$.
\end{proof}


\subsection{Prover-Delayer Games}\label{sec:PD-game}

The \textit{Prover-Delayer game} is an approach to obtain lower bounds on resolution refutations  introduced by  Pudl\'ak and Impagliazzo \cite{PI00}. The idea is that the non-existence of small decision trees, and hence small tree-like resolution refutations, for an unsatisfiable formula,
can be phrased in terms of the existence of a certain
strategy for  Delayer in a  game against  Prover, associated to the unsatisfiable formula. We define such games $G^R$ and $G^R_{sw}$ for
decision trees $\DT{R}$ and $\DTsw{R},$
respectively. Below we show (Lemma~\ref{gamesTrees}) that the existence of certain strategies for the Delayer in $G^R$ and $G^R_{sw}$ imply lower bounds on
the size of $\DT{R}$ and $\DTsw{R}$ trees, respectively.

\para{The game.} Let $\phi$ be a set of linear clauses and $\Phi_s$ be a set of linear non-equalities. Consider the following game between two parties
called Prover and Delayer. The game goes in rounds, consisting of one move of Prover
followed by one move of Delayer. The position in the game is determined by a system of linear non-equalities $\Phi$, which is extended by one non-equality after
every round. The starting position is $\Phi_s$.

In each round, Prover presents to  Delayer a possible branching $f\neq 0$ and  $g\neq 0$ over  
a linear non-equality $f+g\neq 0$, such that $f+g\neq 0\in\Phi\cup\{a\neq 0\,|\,a\in R\setminus 0\}$ or $\Phi\models f+g\neq 0$ in $G^R$ and $G^R_{sw}$,
respectively. After that, Delayer  chooses either $f\neq 0$ or $g\neq 0$ to be added to $\Phi$, or leaves the choice to
the Prover and thus earns a coin. The game $G^R$ finishes, when $\neg C \subseteq \Phi$ for some $C\in\phi\cup\{0=0\},$ and $G^R_{sw}$
finishes, when $\Phi\models \neg C$ for some clause $C\in \phi\cup \{0=0\}$.
\begin{lemma}\label{gamesTrees}
If there exists a strategy with a starting position $\Phi_s$ for Delayer in the game $G^R$ (respectively, $G^R_{sw}$) that guarantees at least $c$ coins on
a set of linear clauses $\phi$, then the size of a $\DT{R}$ (respectively $\DTsw{R}$) tree for $\phi$, with the system $\Phi_s$ in the root, must be at
least $2^c$.
\end{lemma}
\begin{proof}
Assume that  $T$ is a tree of type $\DT{R}$ (respectively, $\DTsw{R}$) for $\phi$. We define an embedding of the full binary tree $B_c$ of height $c$ to $T$ inductively as follows. 
We simulate Prover in the game $G^R$ (respectively, $G^R_{sw}$) by choosing branchings from $T$ and following to a subtree chosen by the Delayer until Delayer decides to earn a
coin and leaves the choice to the Prover or until the game finishes. In case we are at a position where Delayer earns a coin, and which corresponds to a vertex $u$ in $T$, we map 
the root of $B_c$ to $u$ and proceed inductively by embedding two trees $B_{c-1}$ to the left and right subtrees of $u$, corresponding to two choices of the Prover.
\end{proof}

\subsection{Lower Bounds for the Subset Sum with Small Coefficients}

We now turn to tree-like lower bounds. 
In this section we prove \treslin\Q\ lower bound for $\subSum{f}$ including instances, where coefficients of
$f$ are small, and \treslinsw\F\ lower bound for $\negIm{\pm x_1\pm \dots\pm x_n}$.

The proof of \treslin{\Q} lower bound for $\subSum{f}$ goes in two stages. Assume $f$ depends on $n$ variables. First, as in the proof of
dag-like lower bound in Sec.~\ref{sec:LBDL} we use Theorem~\ref{thm:refToDer} to transform refutations $\pi$ of $f=0$ to derivations
$\pi'$ of a clause $C_{\pi}$ from only the boolean  axioms. We ensure that $\pi'$ is not much larger than $\pi$ and $C_{\pi}$ possesses the following
property, which makes it
hard to derive: for every disjunct $g=0$ in $C_{\pi}$ the linear polynomial $g$ depends on at least $\frac{n}{2}$ variables.
Second, we use Prover-Delayer games to prove the lower bound for derivations of any clause with this property. The proof that 
Delayer's strategy succeeds to earn sufficiently many coins is guaranteed by a bound on size of essential coverings of hypercubes.

\begin{definition}
Let $\mathcal{H}$ be a set of hyperplanes in $\Q^n$. We say that $\mathcal{F}$ forms
\textbf{essential cover} of the cube $B_n=\{0,1\}^n$ if:
\begin{itemize}
\item Every point of $B_n$ is covered by some hyperplane in $\mathcal{H}$.
\item No proper subset $\mathcal{H'}\subsetneq \mathcal{H}$ covers $B_n$.
\item No axis in $\Q^n$ is parallel to all hyperplanes in $\mathcal{H}$. In other words, if $\mathcal{H}=\{H_1,\ldots,H_m\}$ and $f_i=0$
is the linear equation defining $H_i$, $i\in[m]$, then every variable $x_j$, $j\in [n]$, occurs with nonzero coefficient in some $f_i$.
\end{itemize}
\end{definition}

\begin{theorem}[\cite{LINIAL05}]\label{thm:essCov}
Any essential cover of the cube $B_n$ in $\Q^n$ must contain at least $\frac{1}{2}(\sqrt{4n+1}+1)$ hyperplanes.
\end{theorem}

We use Prover-Delayer games to prove the lower bound below.

\begin{theorem}\label{thm:largeWeightLB}
Any \treslin{\Q} derivation of any tautology of the form $\bigvee\nolimits_{j\in [N]}g_j=0$, for some positive $N$,  where each $g_j$ is linear over $\Q$ and depends on at least $\frac{n}{2}$
variables, is of size $2^{\Omega(\sqrt{n})}$.
\end{theorem}
\begin{proof}
According to the definitions in Sec.~\ref{sec:PD-game} the corresponding Prover-Delayer game is on $0=0$ and
starts with the position $$\Phi_r={\{g_j\neq 0\;|\;j\in [N]\}}\,.$$
The game finishes at a position $\Phi$, where $\{x_i\neq 0,~ x_i\neq 1\}\subseteq\Phi$ for some $i\in[n]$ or $0\neq 0\in\Phi$.

We now define a Delayer's strategy that guarantees $\Omega(\sqrt{n})$ coins and by Lemma~\ref{gamesTrees} obtain the lower bound. 

If $\Phi$ is a position in the game, denote by $\Phi_c\subset\Phi$ the subset of so-called ``coin'' non-equalities, that is, non-equalities that
were chosen by Prover when Delayer 
decided to leave the choice to Prover and earn a coin. The number $|\Phi_c|$ is then precisely the number of coins earned by Delayer
at $\Phi$. Throughout the game Delayer constructs a partial assignment $\rho_I$ for variables
in $I\subseteq [n]$ and a set of non-equalities $\Phi_I\subseteq\Phi_c$, such that:
\begin{enumerate}
\item 
 $|\Phi_I|=\Omega(\sqrt{|I|})$;
 \item  for all $g\neq 0\in(\Phi\rst_{\rho_I})\setminus(\Phi_c\rst_{\rho_I})$,
the function $g$ depends on at least $\frac{n}{2}-|I|$ variables;
\item  $\Phi_I$ contains variables only from $I$; and
\item  $\Phi_c\rst_{\rho_I}$ is 0-1 satisfiable. 
\end{enumerate}
In the beginning both $\rho_I$ and $\Phi_I$ are empty.

Let the position in the game be defined by a system $\Phi$ and let the branching chosen by the Prover be
$g_1\neq 0$ and $g_2\neq 0$, where $g_1+g_2\neq 0\in\Phi$. Delayer does the following. Before making any decision Delayer checks if 
there exists some nonconstant linear $g$ with variables in $[n]\setminus I$ such that $(\Phi_c\rst_{\rho_I})\cup \{g\neq 0\}$
is unsatisfiable over 0-1. 

In case it holds, $\Psi:=(\Phi_c\setminus\Phi_I)\rst_{\rho_I}\cup\{g\neq 0\}$ must be 0-1 unsatisfiable.
Consider a minimal subset $\Psi'\subseteq\Psi$ such that $\Psi'$ is 0-1 unsatisfiable and denote $I'\subseteq [n]$ the set of variables
that occur in $\Psi'$.  As $\Psi'':=\Psi'\setminus\{g\neq 0\}$ is 0-1 satisfiable, there
exists an assignment $\rho_{I'}$ for variables in $I'$, that satisfies $\Psi''$. Delayer extends the assignment $\rho_I$ with 
$\rho_{I'}$ to $\rho_{I\cup I'}$ and defines $\Phi_{I\cup I'}:=\Phi_I\cup \Psi''$.

If $\Psi'=\{g_1\neq 0, \ldots, g_k\neq 0\}$, then the hyperplanes $H_1,\ldots,H_k$ defined by the equations 
$g_1=0,\ldots,g_k=0$ form an essential cover of the cube $B_{|I'|}$. Therefore, by Theorem~\ref{thm:essCov}, $|\Psi''|=|\Psi'|-1\geq \sqrt{|I'|}$
and thus $|\Phi_{I\cup I'}|\geq \sqrt{|I|}+\sqrt{|I'|}\geq \sqrt{|I\cup I'|}$.

If necessary, Delayer repeats the above procedure constructing extensions $\rho_{I_1}\subset\dots\subset\rho_{I_L}$ and 
$\Phi_{I_1}\subset \dots\subset \Phi_{I_L}$, where $I_1=I\subset\ldots\subset I_L$, until there is no $g\neq 0$ 
inconsistent with $\Phi_c\rst_{\rho_{I_L}}$ as described above. The new value of $I$ is set to $I_L$. After that
Delayer does the following: 

\begin{enumerate}
\item\label{it:g2Case} if $g_1\rst_{\rho_I}=0$, then choose $g_2\neq 0$;
\item\label{it:g1Case} otherwise, if $g_2\rst_{\rho_I}=0$, then choose $g_1\neq 0$;
\item if none of the above cases hold, leave the choice to Prover and earn a coin.
\end{enumerate}

Denote by $\Phi'$ and $\Phi_c'\subseteq \Phi'$ the new position and the subset of ``coin'' non-equalities, respectively, 
after the choice is made. It is easy to see that 
the property that any $g\neq0\in (\Phi'\rst_{\rho_I})\setminus(\Phi_c'\rst_{\rho_I})$ depends on at least $\frac{n}{2}-|I|$
variables still holds.

It follows from the definition of Delayer's strategy that $\Phi_c$ is always 0-1 satisfiable. Therefore 
if $\Phi$ is the endgame position, that is if $0\neq 0\in \Phi$
or $\{x_i\neq 0,x_i\neq 1\}\subset\Phi$ for some $i\in[n]$, then $0\neq 0\in(\Phi\rst_{\rho_I})\setminus(\Phi_c\rst_{\rho_I})$
or $\{x_i\neq 0,x_i\neq 1\}\subset(\Phi\rst_{\rho_I})\setminus(\Phi_c\rst_{\rho_I})$
respectively. This implies that $|I|\geq \frac{n}{2}-1$ and therefore $|\Phi_c|\geq |\Phi_I|\geq \sqrt{|I|}=\Omega(\sqrt{n})$.
Thus the number of coins earned by Delayer is $\Omega(\sqrt{n})$.

\end{proof}

\begin{corollary}\label{cor:imTLLB}
Let $f$ be any linear polynomial over \Q\  that depends on $n$ variables. Then \treslin\Q\ derivations of $\imAx{f}$ are of size $2^{\Omega(\sqrt{n})}$. 
\end{corollary}

\begin{theorem}\label{thm:ssTLLB}
If $f$ is a linear polynomial over $\Q$, which depends on $n$ variables and $0\notin im_2(f)$, then every  \treslin\Q\ refutation
of $f=0$ is of size $2^{\Omega(\sqrt{n})}$.
\end{theorem}
\begin{proof}
Consider the following predicate $\mathcal{P}$ on linear polynomials: $\mathcal{P}(g)=1$ iff $g$ depends on at least $\frac{n}{2}$
variables. It is easy to see that $\mathcal{P}$ satisfies the conditions in Theorem~\ref{thm:refToDer} with respect to $f$. Therefore
by Theorem~\ref{thm:refToDer} for every refutation $\pi$ of $f=0$ there exists a derivation $\pi'$ of a clause $C_{\pi}$ from
the boolean  axioms such that $|\pi'|=O(n\cd |\pi|^3)$ and $\mathcal{P}(g)$ for every $g=0$ in $C_{\pi}$. Thus, by Theorem~\ref{thm:largeWeightLB}
$|\pi'|=2^{\Omega(\sqrt{n})}$ and $|\pi|=2^{\Omega(\sqrt{n})}$.  
\end{proof}

\begin{lemma}\label{immunityLm}
Let $\Phi$ be a satisfiable system of $m$ non-equalities over $\F$.
If $\Phi\models \epsilon_1x_1+\dots+\epsilon_nx_n=A$ for some
$\epsilon_i\in\{-1,1\}\subset{\F},A\in\F$, then $m\geq \frac{n}{4}$.
\end{lemma}
Note that $A$ must be an integer (inside \F), since the coefficients of variables are all $-1,1$, and the variables themselves are boolean  (since $\models$ stands for semantic implication over 0-1 assignments only).
\begin{proof}
Let $\Phi=\{\overline{a}_1\cdot\overline{x}+b_1\neq 0,\dots,\overline{a}_m\cdot\overline{x}+b_m\neq 0\}$ and put $\sigma=A\mod{2}$, 
$f=\epsilon_1x_1+\dots+\epsilon_nx_n$.
Then 
\begin{align*}
f\equiv1-\sigma~(\text{mod}\ 2) & \models f\neq A \\ 
& \models (\overline{a}_1\cdot\overline{x}+b_1)\cdot\ldots\cdot(\overline{a}_m\cdot\overline{x}+b_m)=0.
\end{align*}
By Theorem 4.4 in Alekhnovich-Razborov \cite{AR01}, the function $f\equiv1-\sigma ~(\text{mod}\ 2)$ is $\frac{n}{4}$-\textit{immune}, that is, the degree of any non-zero polynomial $g$ such that 
$f\equiv1-\sigma~(\text{mod}\ 2)\models g=0$ must be at least $\frac{n}{4}$. Therefore $m\geq \frac{n}{4}$.
\end{proof}


\begin{theorem}\label{thm:imAvLB}
We work over \Q. Let $f=\epsilon_1x_1+\dots+\epsilon_nx_n$, where $\epsilon_i\in\{-1,1\}$. Then 
any tree-like \reslinsw{\Q} refutation of $\negIm{f}$ is of size at least $2^{\frac{n}{4}}$.
%
%
\end{theorem}
\begin{proof}
According to the definitions in Sec.~\ref{sec:PD-game} the corresponding Prover-Delayer game is on $\negIm{f}$ and
starts with the empty position.
The game finishes at a position $\Phi$, where $\Phi\models f-A= 0$ for some $A\in im_2(f)$.

We now define a Delayer's strategy that guarantees $\frac{n}{4}$ coins and by Lemma~\ref{gamesTrees} obtain the lower bound. 

The strategy is as follows. Let the position in the game be defined by a system $\Phi$ and let the branching chosen by the Prover be
$g_1\neq 0$ and $g_2\neq 0$, where $\Phi\models g_1+g_2\neq 0$.  
Delayer does the following: 
\begin{enumerate}

\item\label{g1Case}

if $g_2\neq 0$ is inconsistent with $\Phi$, but $g_1\neq 0$ is consistent with $\Phi$, then choose $g_1\neq 0$;

\item\label{g2Case} 
if $g_1\neq 0$ is inconsistent with $\Phi$, but $g_2\neq 0$ is consistent with $\Phi$, then choose $g_2\neq 0$;

\item\label{coinCase} if none of the above holds, then leave the choice to the Prover and earn a coin. 
\end{enumerate}

We now prove that this strategy guarantees the required number of coins.
\medskip

Suppose that the game has finished at a position $\Phi$. The strategy of Delayer guarantees that $\Phi$ is satisfiable and
$\Phi$ contradicts a clause $\langle f\neq A\rangle$ of $\negIm{f}$, that is $\Phi\models f-A=0$ for some $A\in im_2(f)$.
Let $\zeta_1,\ldots,\zeta_\ell$ be the set of non-equalities in $\Phi$, in the order they were added to $\Phi$.
Let $\Psi\subseteq\Phi$ be the set of all $\zeta_i$, $i\in[\ell]$, such that $\zeta_i$ is not implied by previous
non-equalities $\zeta_j$, for $j<i$. Then, Delayer earns at least $|\Psi|$ coins, $\Psi\models f=A$, and by
Lemma~\ref{immunityLm} we conclude that $|\Psi|\geq\frac{n}{4}$.

\medskip

\end{proof}

\subsection{Lower Bounds for the Pigeonhole Principle}

Here we prove that every tree-like \reslinsw{\F} refutations of $\neg\text{PHP}^m_n$ must have size at least $2^{\Omega(\frac{n-1}{2})}$ (see Sec.~\ref{sec:PHP-intro} for the definition of $\neg\text{PHP}^m_n$). Together
with the upper bound for dag-like \reslin{\F} (Theorem \ref{phpUB}) this provides a separation between tree-like and dag-like
\reslinsw{\F} in the case $char(\F)=0$, for formulas in CNF. The lower bound argument is comprised of exhibiting a strategy for Delayer in the Prover-Delayer game.
Delayer's strategy  is similar to that in  \cite{IS14}. However, the proof that Delayer's strategy guarantees sufficiently many  coins relies
on Lemma~\ref{hyplem}, which is a generalization of Lemma 3.3 in \cite{IS14} for arbitrary fields. Since the proof of Lemma 3.3 in \cite{IS14} for the $\F_2$ case does not apply to arbitrary fields, our proof is different, and uses a result from Alon-F\" uredi \cite{AF93} on the hyperplane coverings of the hypercube. 
  
\begin{theorem}\label{phpLB}
For every field $\F$, the shortest \treslinsw{\F} refutation of $\neg\text{PHP}^m_n$ has size $2^{\Omega(\frac{n-1}{2})}$.
\end{theorem}

\begin{proof}
We prove that there exists a strategy for Delayer in the $\neg\text{PHP}^m_n$ game, which guarantees Delayer to earn $\frac{n-1}{2}$ coins. Following the terminology in \cite{IS14}, we call an assignment  $x_{i,j}\mapsto \alpha_{ij}$, for $\alpha \in \{0,1\}^{mn}$, \emph{proper}
if it does not violate $\phpInj^m_n$, namely, if it does not send two distinct pigeons to the same hole. We need to prove several lemmas before concluding the theorem.

\begin{lemma}\label{hyplem}
Let $A\overline{x} \doteqdot \overline{b}$ be a system of $k$ linear non-equalities over a field $\F$ with $n$ variables and where  $\overline{x}=0$ is a solution,
that is, $0\doteqdot \overline{b}$. If $k<n$, then there exists a non-zero boolean  solution to this system.
\end{lemma}
\begin{proof}
Let $\overline{a}_1,\dots,\overline{a}_k$ be the rows of the matrix $A$. The boolean  solutions to the system $A\overline{x} \doteqdot \overline{b}$ are all the points of the
$n$-dimensional boolean  hypercube $B_n:=\{0,1\}^n\subset\F^n$, that are not covered by the hyperplanes $H:=\{\overline{a}_1\overline{x}-b_1=0,\dots,\overline{a}_k\overline{x}-b_k=0\}$.
We need to show that if $k<n$ and $0\in B_n$ is not covered by $H$, then some other point in $B_n$ is not covered by $H$ as well. This follows from  \cite{AF93}:
\newtheorem*{AF93}{Corollary from Alon-F\"{u}redi \cite[Theorem 4]{AF93}}
\begin{AF93} Let  $Y(l):=\left\{(y_1,\dots,y_n)\in\F^n\;|\;\forall i\in[n], 0<y_i\leq 2, \text{ and }\sum_{i=1}^n y_i\geq l\right\}.$
For any field $\F$, if $k$ hyperplanes in $\F^n$ do not cover $B_n$ completely, then they do not cover at least $M(2n-k)$ points from $B_n$, where 
\begin{equation*}
M(l):=\min\limits_{(y_1,\ldots,y_n)\in Y(l)}\prod\limits_{1\leq i\leq n}y_i\,.
\end{equation*}
\end{AF93}
Thus, if $k<n$ hyperplanes do not cover $B_n$ completely, then they do not cover at least $M(n+1)$ points. The set $Y(n+1)$ in the Corollary above consists of all tuples $(y_1,\dots,y_n)$,
where $y_i=2$ for some $i\in[n]$ and $y_j=1$ for $j\in[n],j\neq i$. Therefore $M(n+1)=2$.
\end{proof}
For two boolean  assignments $\alpha,\beta\in\bits^n$, denote by $\alpha\oplus\beta$\ the bitwise \textsc{xor} of the two assignments. 
\begin{lemma}\label{hypcor}
Let $A\overline{x} \doteqdot \overline{b}$ be a system of $k$ linear non-equalities over a field $\F$ with $n>k$ variables and let  $\alpha\in\bits^n$ be a solution to the system. Then, for every choice $I$ of $k+1$ bits in $\alpha$, there exists at least one $i\in I$ so that flipping the $i$th bit in $\alpha$ results  in a new solution to $A\overline{x} \doteqdot \overline{b}$. In other words, if $I\subseteq [n]$ is such that $|I|=k+1$, then there exists a boolean  assignment $\beta\neq 0$ such that $\{i\;|\;\beta_{i}=1\}\subseteq I$ and $A(\alpha\oplus\beta)\doteqdot\overline{b}$. \end{lemma}
\begin{proof}
Let $I\subseteq\bits^n$. Denote by $A^\star_I$ the matrix with columns $\{(1-2\alpha_i)\overline{a}_i\;|\;i\in I\}$, where $\overline{a}_i$ is the $i$th \textit{column} of $A$. That is, $A^\star_I$ is the matrix $ A$ restricted to columns $i$ with $i\in I$ and where column $i$ flips its sign iff $\alpha_i$ is $1$.

Assume that $\beta\in\bits^n$ is nonzero and all its 1's must appear in the indices in $I$, that is, $\{i\;|\;\beta_{i}=1\}\subseteq I$.  Given a set of indices $J\subseteq [n]$, denote by $\beta_J$ the restriction of
$\beta$ to the indices in $J$. Similarly, for a vector $v\in\F^n$,  $v_J$ denotes the restriction of $v$ to the indices in $J$.
\begin{clm*} $A(\alpha\oplus\beta)\doteqdot\overline{b}$ iff $A^\star_I\beta_I\doteqdot \overline{b}-A\alpha$. 
\end{clm*}
\begin{proofclaim}
We prove that $A(\alpha\oplus\beta)=A^\star_I\beta_I+A\alpha $. Consider any row $\bf v$ in $A$, and the corresponding row ${\bf v}^\star_I$ in $ A^\star_I$. Notice that  $\bf v\cd (\alpha\oplus\beta)$ (for ``$\cd$'' the dot product) equals 
the dot product of $\bf v$ and $\alpha\oplus\beta$, where both vectors are restricted only to those entries in which $\alpha$ and $\beta$ differ.
Considering entries outside $I$, by assumption we have $\beta_{[n]\setminus I}=0$, which implies that 
\begin{equation}\label{eq:stam_efes}
{\bf v}_{[n]\setminus I}\cd(\alpha\oplus\beta)_{[n]\setminus I}={\bf v}_{[n]\setminus I}\cd \alpha_{[n]\setminus I} \,.
\end{equation}  
On the other hand, considering entries inside  $I$, we have 
\begin{equation}\label{eq:stam}
{\bf v}_I\cd (\alpha\oplus\beta)_{I}= {\bf v}_I\cd \alpha_I +  {\bf v}^\star_{I}\cd \beta_I\,.
\end{equation}
Equation \eqref{eq:stam} can be verified by inspecting all four cases for the $i$th bits in $\alpha,\beta$, for  $i\in I$, as follows:  for those indices $i\in I$, such that $\alpha_i=1$ and $\beta_i =0$,  only ${\bf v}_I\cd\alpha$ contributes to the right hand side in \eqref{eq:stam}. If  $\alpha_i=1$ and $\beta_i =1$, then by the definition of $A^\star_I$, the two summands in the right hand side in \eqref{eq:stam} cancel out. The cases  $\alpha_i=0, \beta_i =1$ and $\alpha_i=\beta_i =0$, can also be inspected to contribute the same values to both sides of \eqref{eq:stam}.

The two equations \eqref{eq:stam_efes} and \eqref{eq:stam}  concludes the claim.   
\end{proofclaim}

We know that $A\alpha\doteqdot \overline  b$, and we wish to show that for some nonzero $\beta\in\bits^n$ where $\{i\;|\;\beta_{i}=1\}\subseteq I$, it holds that  $A(\alpha\oplus\beta)\doteqdot\overline b$.  By the claim above it remains to show the existence of such $\beta$ where $A^\star_I\beta_I\doteqdot \overline{b}-A\alpha$. But notice that $\overline{b}-A\alpha\doteqdot 0$, since $A\alpha\doteqdot \overline b$, and that $A^\star_I\beta_I$ is a matrix of dimension $k\times(k+1)$. Therefore, by Lemma \ref{hyplem}, the system $A^\star_I\beta_I\doteqdot \overline{b}-A\alpha$ has a nonzero solution, that is, there exists a $\beta\neq 0$ for which all ones are in the $I$ entries, such that $A^\star_I\beta_I\doteqdot \overline{b}-A\alpha$.
\end{proof}

\begin{lemma}\label{propsollem}
Assume that a system $A\overline{x} \doteqdot \overline{b}$ of $k\leq \frac{n-1}{2}$ non-equalities over $\F$ with variables $\{x_{i,j}\}_{(i,j)\in [m]\times[n]}$
has a proper solution. Then, for every $i\in [m]$ there exists a proper solution to the system, that  satisfies the clause $\bigvee\nolimits_{j\in [n]}x_{i,j}$. In other words, for every pigeon, there exists a proper solution that sends the pigeon to some hole.    
\end{lemma}
\begin{proof}

We first show that if there exists a proper solution of $A\overline{x} \doteqdot \overline{b}$, then there exists a proper solution of this system with at most $k$
ones. Let $\alpha$ be a proper solution with at least $k+1$ ones. If $I$ is a subset of $k+1$ ones in $\alpha$, then Lemma~\ref{hypcor} assures us that some
other proper solution can be obtained from $\alpha$ by flipping some of these ones (note that flipping one to zero preserves the properness of assignments). Thus the number of ones can always be reduced until it is at most $k$.

Let $\alpha$ be a proper solution with at most $k$ ones. The condition $k\leq \frac{n-1}{2}$ implies that there are $n-k\geq k+1$ free holes. Let $J$ be a subset
of size $k+1$ of the set of indices of free holes. Then for any $i\in [m]$ some of the bits in $I=\{(i,j)\;|\;j\in J\}$ can be flipped and still satisfy $A\overline{x} \doteqdot \overline{b}$, by Lemma~\ref{hypcor}. (As before, flipping from one to zero maintains the properness of the solution.) Hence, the resulting
proper solution must satisfy the clause $\bigvee\nolimits_{j\in [n]}x_{i,j}$.
\end{proof}

We now describe the desired strategy for Delayer. \smallskip

\noindent\uline{Delayer's Strategy}: Let a position in the game be defined by the system of non-equalities $\Phi$ and assume that the branching
chosen by Prover is
$f_0\neq 0 $ or $f_1\neq 0$, where $\Phi\models f_0+f_1\neq 0$. The only  objective of Delayer is to ensure that the system $\Phi$ has proper solutions.
Delayer uses the opportunity to earn a coin whenever both $\Phi\cup \{f_0\neq 0\}$ and $\Phi\cup \{f_1\neq 0\}$ have proper
solutions by leaving the choice to Prover. Otherwise, in case $\Phi\wedge \phpInj^m_n\models f_i= 0$, for some $i\in\bits$, Delayer chooses
$f_{1-i}\neq 0$, which must satisfy $\Phi\wedge \phpInj^m_n\models f_{1-i}\neq 0$, and so the sets of proper
solutions of $\Phi$ and $\Phi\cup \{f_{1-i}\neq 0\}$ are identical.\smallskip \smallskip

This strategy ensures, that for  every end-game position $\Phi$, $\Phi$ has proper
solutions and $\Phi\models\neg\phpTot^m_n$. Note  that $\Phi$ has the same proper solutions as $\Phi'$, obtained by throwing away from $\Phi$ all non-equalities
that were added by Delayer when  making a choice. Therefore, if $\Phi\models\neg\phpTot^m_n$, then $\Phi'\wedge \phpInj^m_n\models\neg\phpTot^m_n$
and thus $|\Phi'|>\frac{n-1}{2}$ by Lemma~\ref{propsollem}. 

Since $|\Phi'|$ is precisely the number of coins earned by Delayer, this gives the desired lower bound. 
\end{proof}


\section{Size-Width Relation and Simulation by Polynomial Calculus}

In this section we prove a size-width relation for tree-like \reslin{R} (Theorem~\ref{sizeWidth}), which then implies an exponential lower bound 
on the size of \treslinsw{R} refutations in terms of the principal width of refutations (Definition~\ref{omegaDef}). The connection between the principal
width and the degree of PC refutations for finite fields \F, together with lower bounds on degree of PC refutations from \cite{AR01} on Tseitin mod $p$ formulas and random CNFs, imply exponential lower bounds for the size of
\treslinsw{\F} for these instances (Corollaries \ref{tsLB} and \ref{rndLB}).



\begin{proposition}\label{substProp}
Let $\phi=\{C_i\}_{1\leq i\leq m}$ be a set of linear clauses and $x\in vars(\phi)$. Assume that $l$ is a linear form in the
variables $vars(\phi)\setminus\{x\}$. Then, there is a \reslin{R} derivation $\pi$ of 
$\{C_i \rst_{x\leftarrow l}\vee \langle x-l\neq 0\rangle\}_{1\leq i\leq m}$ from $\phi$ of size polynomial in $|\phi|+|\imAx{l}|$
and such that $\omega_0(\pi)\leq \omega_0(\phi)+2$. 
\end{proposition}
\begin{proof}
The clause $x-l=0\vee \langle x-l\neq 0\rangle$ is derivable in \reslin{R} in polynomial in $|\imAx{l}|$ size by Proposition~\ref{imf}.  Assume 
$$C=\left(\bigvee\nolimits_{j\in[k]}f_j+a_jx+b_j^{(1)}=0\vee\dots\vee f_j+a_jx+b_j^{(N_j)}=0\right),$$ where $x\notin vars(f_i)$ and
we have grouped disjuncts so that $\omega_0(C)=k$.
Then we resolve these groups one by one with $x-l=0\vee \langle x-l\neq 0\rangle$ and after $N_1+\ldots+N_k$ steps yield 
$\left(\bigvee\nolimits_{j\in[k]}f_j+a_jl+b_j^{(1)}=0\vee\dots\vee f_j+a_jl+b_j^{(N_j)}=0\vee \langle x-l\neq 0\rangle\right)$. It is
easy to see that the principal width never exceeds $k+2$ along the way. Therefore $\omega_0(\pi)\leq \omega_0(\phi)+2$. 
\end{proof}

\begin{corollary}\label{substCor}
Let $\phi=\{C_i\}_{1\leq i\leq m}$ be a set of linear clauses and $x\in vars(\phi)$. Suppose that $l$ is a linear form with variables
 $vars(\phi)\setminus\{x\}$ and that $\pi$ is a \reslin{R} refutation of $\phi\rst_{x\leftarrow l}\cup\{l=0\vee l=1\}$.
Then, there exists a \reslin{R} derivation $\widehat{\pi}$ of
${\langle x-l\neq 0\rangle}$ from $\phi$, such that $S(\widehat{\pi})=O(S(\pi)+|\imAx{l}|)$ and 
$\omega_0(\widehat{\pi})\leq\max\left(\omega_0(\pi)+1,\omega_0(\phi)+2\right)$. Additionally,
there is a refutation $\widehat{\pi}'$ of $\phi\cup \{x-l=0\}$ where $\omega_0(\widehat{\pi}')\leq
\max(\omega_0(\pi),\omega_0(\phi)+2)$.
\end{corollary}
\begin{proof}

By Proposition~\ref{substProp} there exists a derivation $\pi_s$ of 
$$
\{C_i\rst_{x\leftarrow l}\vee \langle x-l\neq 0\rangle\}_{1\leq i\leq m}\cup\{l=0\vee l=1\vee \langle x-l\neq 0\rangle\}
$$
from $\phi$ of width at most $\omega_0(\phi)+2$. Composing $\pi_s$ with $\pi\vee \langle x-l\neq 0\rangle$ yields the derivation
$\widehat{\pi}$ of $\langle x-l\neq 0\rangle$ from $\phi$. 

Moreover, by taking the derivation $\pi_s$ and adding to it the axiom $x-l=0$, and then using a sequence of resolutions of $\pi_s$
with $x-l=0$, we obtain a derivation of
$\phi\rst_{x\leftarrow l}\cup\{l=0\vee l=1\}$ from $\phi\cup\{x-l=0\}$. The latter derivation composed with $\pi$ yields the refutation
$\widehat{\pi}'$ of $\phi\cup\{x-l=0\}$ of width at most $\max(\omega_0(\pi),\omega_0(\phi)+2)$.
\end{proof}

\begin{theorem}\label{sizeWidth}
Let $\phi$ be an unsatisfiable set of linear clauses over a field \F. The following size-width relation holds for both \treslin{\F} and \treslinsw{\F}:
$$
{S(\phi\vdash\perp)=2^{\Omega(\omega_0(\phi\vdash\perp)-\omega_0(\phi))}}\,.
$$
\end{theorem}
\begin{proof}

We prove by induction on $n$, the number of variables in $\phi$, the following:
$$
\omega_0(\phi\vdash\perp)\leq \lceil \log_2 S(\phi\vdash\perp)\rceil+\omega_0(\phi)+2\,.
$$ 

\Base $n=0$. Thus $\phi$ must contain only linear clauses $a=0$, for $a\in \F$, and the principal width for refuting $\phi$ is therefore 1.
 
\induction
Let $\pi$ be a tree-like refutation of $\phi=\{C_1,\dots,C_m\}$ such that $S(\pi)=S(\phi\vdash\perp)$ (i.e., $\pi$ is of minimal size).
Without loss of generality, we assume that the resolution rule in $\pi$ is only applied to simplified clauses, that is clauses not containing
disjuncts $1=0$ in case of \treslin{\F} and not containing unsatisfiable $f=0,\,0\notin im_2(f)$ in case of \treslinsw{\F}. The
former can be eliminated by the simplification rule and the latter by the semantic weakening rule.
By this assumption, the empty clause at the root of $\pi$ is derived in \treslin{\F} (resp.~\treslinsw{\F}) as a simplification (resp.~weakening)
of an unsatisfiable $h=0$ ($1=0$ in case of \treslin{\F}) equation, which is derived by application of the resolution rule. Denote the left and right subtrees,
corresponding to the premises of $h=0$, by $\pi_1$ and $\pi_2$, respectively.

The roots of $\pi_1$ and $\pi_2$ must be of the form $f_1=0$ and $f_2=0$, respectively, where $f_1-f_2=h$. 
Therefore, $$f_1=l(x_1,\dots,x_{n-1})+a_nx_n
\text{ ~and~ }
f_2=l(x_1,\dots,x_{n-1})+a_nx_n-h\,,$$ for some  $l(x_1,\dots,x_{n-1})=\sum\nolimits_{i=1}^{n-1}a_ix_i+B$, where $a_i,B\in \F$. 

Assume without loss of generality that $a_n\neq 0$ and $S(\pi_1)\leq S(\pi_2)$. We now use the induction hypothesis to construct a narrow derivation $\pi_1^{\bullet}$ of $f_1=0$ such that
\begin{align*}
\omega_0(\pi_1^{\bullet}) & \leq \lceil \log_2 S(\pi_1)\rceil+1+\omega_0(\phi)+2 \\
& \leq \lceil \log_2 S(\pi)\rceil+\omega_0(\phi)+2\,.
\end{align*}

For every nonzero $A\in im_2(f_1)$ define the partial linear substitution
$\rho_A$ as $x_n\leftarrow (A-l(x_1,\dots,x_{n-1}))a_n^{-1}$. Thus, $f_1\rst\rho_A = A$. The set of linear clauses  
\begin{equation}\label{eq:mashehu}
\phi\rst_{\rho_A}\cup\left\{{(A-l)a_n^{-1}=0}\vee {(A-l)a_n^{-1}=1} \right
\}\end{equation}
is  unsatisfiable and has  $n-1$ variables, and is refuted by $\pi_1\rst_{\rho_A}$. 

By induction hypothesis there exists a (narrow) refutation  $\pi_1^A$ of \eqref{eq:mashehu} with \begin{align*}
\omega_0(\pi_1^A)&\leq \lceil \log_2 S(\pi_1\rst_{\rho_A})\rceil+\omega_0(\phi)+2 \\
& \leq \lceil \log_2 S(\pi_1)\rceil+\omega_0(\phi)+2\,.
\end{align*}
By Corollary~\ref{substCor}
there exists a derivation $\widehat{\pi}_1^A$ of $\langle l+a_nx_n\neq A\rangle$ from $\phi$ such that
$\omega_0(\widehat{\pi}_1^A)\leq \max(\omega_0(\pi_1^A)+1,\omega_0(\phi)+2)\leq \lceil \log_2 S(\pi_1)\rceil+\omega_0(\phi)+3$. By Proposition~\ref{eqFromNeq}
there exists 
a derivation $\pi_1^{\bullet}$ of $f_1=0$ such that $\omega_0(\pi_1^{\bullet})\leq \lceil \log_2 S(\pi_1)\rceil+\omega_0(\phi)+3\leq \lceil \log_2 S(\pi)\rceil+\omega_0(\phi)+2$.

Consider the following substitution $\rho$: $x_n \leftarrow -l\cdot a_n^{-1}$. Then, $\pi_2|_{\rho}$ is a derivation of $h=0$ from
$\phi|_{\rho}\cup\{{-l\cdot a_n^{-1}=0}\vee {-l\cdot a_n^{-1}=1}\}$, which we augment to refutation $\pi'_2$ by taking composition with simplification (resp.~weakening) in case of \treslin{\F} (resp. \treslinsw{\F}). By induction hypothesis there exists a refutation $\pi_2^{\bullet}$ of width 
\begin{align*}
\omega_0(\pi_2^{\bullet}) &\leq \lceil \log_2 (S(\pi'_2)+1)\rceil+\omega_0(\phi)+2 \\ &
\leq \lceil \log_2 S(\pi)\rceil+\omega_0(\phi)+2\,,
\end{align*}
and thus by Corollary~\ref{substCor} there exists a refutation $\widehat{\pi}_2^{\bullet}$ of $\phi\cup\{f_1=0\}$ of width 
$\omega_0(\widehat{\pi}_2^{\bullet})\leq \lceil \log_2 S(\pi)\rceil+\omega_0(\phi)+2$. The combination of $\widehat{\pi}_2^{\bullet}$ and $\pi_1^{\bullet}$
gives a refutation of $\phi$ of the desired width.
\end{proof}

\begin{theorem}\label{pcsim}
Let $\F$ be a field and $\pi$ be a \reslin{\F} refutation of an unsatisfiable set of linear clauses $\phi$. Then, there exists a \PCa{\F} refutation
$\pi'$ of (the arithmetization of) $\phi$ of degree $\omega(\pi)$.
\end{theorem}
\vspace{-12pt} 
\begin{proof}
The idea is to replace every clause $C=(f_1=0\vee\ldots\vee f_m=0)$ in $\pi$ by its arithmetization $a(C):=f_1\cd\ldots\cd f_m$, and then augment this sequence to a valid \PCa{\F} derivation by simulating all the rule applications in $\pi$ by several \PCa{\F} rule applications.  

\case 1 If $D=(C\vee g_1=0\vee\ldots\vee g_m=0)$ is a weakening of $C$, then apply the product and the addition rules to derive $a(D)=a(C)\cd g_1\cd\ldots\cd g_m$
from $a(C)$.

\case 2 If $D$ is a simplification of $D\vee 1=0$, then $a(D)=a(D\vee 1=0)$.

\case 3 If $D=(x=0\vee x=1)$ is a a boolean  axiom, then $a(D)=x^2-x$ is an axiom of \PCa{\F}.

\case 4 If $D=(C\vee C'\vee E\vee \alpha f+\beta g=0)$ is a result of resolution of $(C\vee E\vee f=0)$ and $(C'\vee E\vee g=0)$, where $C$ and
$C'$ do not contain the same disjuncts, then by the product and addition rules of PC we derive $a(C)\cd a(C')\cd a(E)\cd f$ from  $a(C\vee E\vee f=0)=a(C)\cd a(E)\cd f$, and also derive $a(C)\cd a(C')\cd a(E)\cd g$ ~from $~a(C'\vee E\vee f=0)=a(C')\cd a(E)\cd f$, and then apply the addition rule to derive $a(C)\cd a(C')\cd a(E)\cd (\alpha f+\beta g)=a(D)$.

It is easy to see that the degree of the resulting \PCa{\F} refutation is at most $\omega(\pi)$.
\end{proof}

As a consequence of Theorems \ref{sizeWidth} and \ref{pcsim}, and the relation $\omega_0\geq \frac{1}{|\F|}\omega$ as well as the results from \cite{AR01}, we have the following:

\fedor{Actually, as of now, we only have lower bounds on width of \reslin{\F} refutations (not on \reslinsw{\F} refutations).
It should not be hard to see, whether \reslinsw{\F} width is the same as \reslin{\F} width.}

\begin{corollary}\label{tsLB}
For every prime  $p$ there exists a constant $d_0=d_0(p)$ such that the following holds. If $d\geq d_0$, $G$ is a $d$-regular Ramanujan
graph on $n$ vertices (augmented with arbitrary orientation to its edges) and $\F$ is a finite field with  $char(\F)\neq p$, then for every
function $\sigma$ such that $\neg\text{TS}^{(p)}_{G,\sigma}\in \text{UNSAT}$, every \treslin{\F} refutation of $\neg\text{TS}^{(p)}_{G,\sigma}$
has size $2^{\Omega(dn)}$.
\end{corollary}
\begin{proof}
Corollary 4.5 from \cite{AR01} states that the degree of \PCa{\F} refutations of $\neg\text{TS}^{(p)}_{G,\sigma}$ is $\Omega(dn)$. Theorem~\ref{pcsim}
implies that the principal width of \reslin{\F} refutations of $\neg\text{TS}^{(p)}_{G,\sigma}$ is $\Omega(\frac{1}{|\F|}dn)=\Omega(dn)$ and thus
by Theorem~\ref{sizeWidth} the size is $2^{\Omega(dn)}$.
\end{proof}

\begin{corollary}\label{rndLB}
Let $\phi\sim\mathcal{F}^{n,\Delta}_k,k\geq 3$ and $\Delta=\Delta(n)$ be such that $\Delta=o(n^{\frac{k-2}{2}})$ and let \F\ be any finite field. Then every \treslin{\F} refutation of $\phi$ has size $2^{\Omega\left(\frac{n}{\Delta^{2/(k-2)}\cdot\log{\Delta}}\right)}$ with probability $1-o(1)$.
\end{corollary}
\begin{proof}
Corollary 4.7 from \cite{AR01} states that the degree of \PCa{\F} refutations of $\phi\sim\mathcal{F}^{n,\Delta}_k$, where $k\geq 3$, is $\Omega(dn)$ with
probability $1-o(1)$. Theorem~\ref{pcsim}
implies that the principal width of \reslin{\F} refutations of $\phi\sim\mathcal{F}^{n,\Delta}_k$ is $\Omega(\frac{1}{|\F|}dn)=\Omega(dn)$ and thus
by Theorem~\ref{sizeWidth} the size of the refutations is $2^{\Omega(dn)}$ with probability $1-o(1)$.
\end{proof}
 
\section*{Acknowledgments}
We wish to thank Dima Itsykson and Dima Sokolov for very helpful comments concerning this work, and telling us about the lower bound on random $k$-CNF formulas for tree-like \reslin{\F_2} that can be achieved using the results of Garlik and Ko\l odziejczyk. We thank Edward Hirsch for spotting a gap in the initial proof of the dag-like lower bound concerning the use of the weakening rule. 


%
\newcommand{\STOC}{STOC}

\small  
\bibliographystyle{plain}
\bibliography{PrfCmplx-Bakoma}
\normalsize


\end{document}